%% file: Main.tex
\documentclass[11pt, final, onecolumn]{article}
\usepackage[backend=bibtex,style=alphabetic,maxbibnames=99,maxalphanames=3,natbib=true]{biblatex}
\usepackage{amsmath, amsfonts, amsthm, graphicx, enumitem, titlesec, comment, appendix}
\usepackage{amssymb}
\usepackage{thmtools, thm-restate}
\usepackage[font=footnotesize]{caption}
\usepackage{sidecap}
\allowdisplaybreaks
\titlelabel{\thetitle.\quad}
\usepackage[margin=1in]{geometry}

\newcommand{\vect}[1]{\mathbf{#1}}

\newcommand{\RR}{\mathbb{R}}

\newcommand\norm[1]{\left\lVert#1\right\rVert}
\newcommand\abs[1]{\left\lvert#1\right\rvert}

\newcommand\parens[1]{\left(#1\right)}

\newcommand\angles[1]{\langle#1\rangle}
\newcommand*{\pr}[2][]{\text{Pr}\ifx\\\left[#1\right]\\\else_{#1}\fi \left[#2\right]}
\newcommand*{\EE}[2][]{\mathbb{E}\ifx\\\left[#1\right]\\\else_{#1}\fi \left[#2\right]}
\newcommand{\partl}[2]{\frac{\partial #1}{\partial #2}}

\newcommand{\prb}{\text{pr}}
\newcommand{\wt}{\text{wt}}
\newcommand{\ws}{\text{WS}}

\newtheorem{ctr}{ctr}[section]
\newtheorem{theorem}[ctr]{Theorem}

\newtheorem{prop}[ctr]{Proposition}
\newtheorem{corollary}[ctr]{Corollary}
\newtheorem{claim}[ctr]{Claim}
\newtheorem{defin}[ctr]{Definition}
\newtheorem*{nonothm}{Theorem}

\theoremstyle{definition}
\newtheorem{myalgorithm}[ctr]{Algorithm}
\newtheorem{remark}[ctr]{Remark}

\title{From Proper Scoring Rules to Max-Min Optimal Forecast Aggregation}

\author{Eric Neyman, Tim Roughgarden}

\begin{document}
\begin{titlepage}
\maketitle

\begin{abstract}
This paper forges a strong connection between two seemingly unrelated forecasting problems: incentive-compatible forecast elicitation and forecast aggregation.  Proper scoring rules are the well-known solution to the former problem.  To each such rule $s$ we associate a corresponding method of aggregation, mapping expert forecasts and expert weights to a ``consensus forecast," which we call \emph{quasi-arithmetic (QA) pooling} with respect to $s$.  We justify this correspondence in several ways:

\begin{itemize}
\item QA pooling with respect to the two most well-studied scoring rules (quadratic and logarithmic) corresponds to the two most well-studied forecast aggregation methods (linear and logarithmic). 

\item Given a scoring rule $s$ used for payment, a forecaster agent who sub-contracts several experts, paying them in proportion to their weights, is best off aggregating the experts' reports using QA pooling with respect to $s$, meaning this strategy maximizes its worst-case profit (over the possible outcomes). 

\item The score of an aggregator who uses QA pooling is concave in the experts' weights. As a consequence, online gradient descent can be used to learn appropriate expert weights from repeated experiments with low regret. 

\item The class of all QA pooling methods is characterized by a natural set of axioms (generalizing classical work by Kolmogorov on quasi-arithmetic means).
\end{itemize}
\end{abstract}
\end{titlepage}

\input{1_Introduction.tex}
\input{2_Related_Work.tex}
\input{3_Preliminaries.tex}
\input{4_Max_Min.tex}

\input{5_Convex_Optimization.tex}
\input{6_Axiomatization.tex}
\input{8_Conclusion.tex}

\section*{Acknowledgments}
This work was partially supported by the NSF under CCF-1813188 and DGE-2036197 and by the ARO under W911NF1910294.

\printbibliography

\begin{appendix}
	\input{Appendix_Intro.tex}
	\input{Appendix_Convex_Losses.tex}
	\input{Appendix_Axioms.tex}
	\input{Appendix_Convex_Exposure.tex}
\end{appendix}
\end{document}

%% file: 1_Introduction.tex
\section{Introduction and motivation} \label{sec:intro}
\subsection{Probabilistic opinion pooling} \label{subsec:pooling}
You are a meteorologist tasked with advising the governor of Florida on hurricane preparations. A hurricane is threatening to make landfall in Miami, and the governor needs to decide whether to order a mass evacuation. The governor asks you what the likelihood is of a direct hit, so you decide to consult several weather models at your disposal. These models all give you different answers: 10\%, 25\%, 70\%. You trust the models equally, but your job is to come up with one number for the governor --- your best guess, all things considered. What is the most sensible way for you to aggregate these numbers?\\

This is one of many applications of \emph{probabilistic opinion pooling}. The problem of probabilistic opinion pooling (or \emph{forecast aggregation}) asks: how should you aggregate several probabilities, or probability distributions, into one? This question is relevant in nearly every domain involving probabilities or risks: meteorology, national security, climate science, epidemiology, and economic policy, to name a few.\\

For a different example of opinion pooling, suppose that you are a consultant tasked with determining the likelihood of your party winning various races in the upcoming election, as part of helping optimally allocate resources. There are many election forecasting models, such as those produced by FiveThirtyEight and The Economist. Again, these models give somewhat different probabilities; what is the most sensible way to aggregate these numbers into a single probability for each race?

In some sense this example is isomorphic to the previous one, but there are domain-specific considerations which might lead to different pooling methods being desirable in the two settings. We will elaborate on this later.\\

As a third example, consider websites that rely on the ``wisdom of crowds" to arrive at a reasonable forecast. One such website is \url{metaculus.com}, which elicits probabilistic forecasts from users on a variety of subjects. In this setting, the goal is to produce an aggregate forecast that reflects the opinion of a large, moderately informed crowd, as opposed to a small, well-informed pool of experts.\\

All of these examples fall into the following framework: there are $m$ experts, who report probability distributions $\vect{p}_1, \dots, \vect{p}_m$ over $n$ possible outcomes (we call these \emph{reports}, or \emph{forecasts}). Additionally, each expert $i$ has a non-negative weight $w_i$ (with weights adding to $1$); this weight represents the expert's quality, i.e. how much the aggregator trusts the expert. A pooling method takes these distributions and weights as input and outputs a single distribution $\vect{p}$. (Where do these weights come from? How can one learn weights for experts? More on this later.)\\

\textbf{Linear pooling} is arguably the simplest of all reasonable pooling methods: a weighted arithmetic mean of the probability distributions.
\[\vect{p} = \sum_{i = 1}^m w_i \vect{p}_i\]
Linear pooling is the most frequently used opinion pooling method across a wide range of applications, including meteorology, economics, and medical diagnoses \cite{rg10}. Indeed, linear pooling frequently outperforms attempts at more sophisticated pooling methods, which often lead to overfitting \cite{wglj18}.\\

\textbf{Logarithmic pooling} (sometimes called \emph{log-linear} or \emph{geometric} pooling) consists of taking a weighted geometric mean of the probabilities and scaling appropriately.
\[p(j) = c \prod_{i = 1}^m (p_i(j))^{w_i}.\]
Here, $p(j)$ denotes the probability of the $j$-th outcome and $c$ is a normalizing constant to make the probabilities add to $1$. The logarithmic pool has a natural interpretation as an average the experts' Bayesian evidence (see Appendix~\ref{appx:intro} for details).\\

The linear and logarithmic pooling methods are by far the two most studied ones, see e.g. \cite{gz86}, \cite{pr00}, \cite{kf08}. This is because they are simple and follow certain natural rules, which we briefly discuss in Section~\ref{sec:related_work}. Furthermore, they are each optimal according to some natural optimality metrics, see e.g. \cite{abbas09}.

\subsection{Proper scoring rules} \label{subsec:proper}
A seemingly unrelated topic within probabilistic forecasting is the truthful elicitation of forecasts: how can a principal structure a contract so as to elicit an expert's probability distribution in a way that incentivizes truthful reporting? This is usually done using a \emph{proper scoring rule}.

A \emph{scoring rule} is a function $s$ that takes as input (1) a probability distribution over $n$ outcomes\footnote{We specify the domain of $s$ more precisely in Section~\ref{sec:prelims}.} and (2) a particular outcome, and assigns a \emph{score}, or reward. The interpretation is that if the expert reports a distribution $\vect{p}$ and event $j$ comes to pass, then the expert receives reward $s(\vect{p}; j)$ from the principal. A scoring rule is called \emph{proper} if the expert's expected score is strictly maximized by reporting their probability distribution truthfully. That is, $s$ is proper if
\[\sum_{j = 1}^n \vect{p}(j) s(\vect{p}; j) \ge \sum_{j = 1}^n \vect{p}(j) s(\vect{x}; j)\]
for all $\vect{x}$, with equality only for $\vect{x} = \vect{p}$. It is worth noting that properness is preserved under positive affine transformations. That is, if $s$ is proper, then $s'(\vect{p}; j) := a s(\vect{p}; j) + b$ is proper if $a > 0$.

\paragraph{\textbf{Quadratic scoring rule}} One example of a proper scoring rule is \emph{Brier's quadratic scoring rule}, introduced in \cite{brier50}. It is given by
\[s_{\text{quad}}(\vect{p}; j) := 2 p(j) - \sum_{k = 1}^n p(k)^2.\]
The quadratic scoring rule can be interpreted as penalizing the expert by an amount equal to the squared distance from their report $\vect{p}$ to the ``true answer" $\delta_j$ (i.e. the vector with a $1$ in the $j$-th position and zeros elsewhere).

\paragraph{\textbf{Logarithmic scoring rule}} Another example of a proper scoring rule is the \emph{logarithmic scoring rule}, introduced in \cite{good52}. It is given by
\[s_{\log}(\vect{p}; j) := \ln p(j).\]
The logarithmic rule is the only proper scoring rule for which an expert's score only depends on the probability assigned to the eventual outcome and not other outcomes \cite{sam66}. The quadratic and logarithmic scoring rules are by far the most studied and most frequently used ones in practice.

\paragraph{\textbf{Choice of scoring rule as a value judgment}} There are infinitely many proper scoring rules. How might a principal go about deciding which one to use? To gain some intuition, we will take a closer look at the quadratic and logarithmic scoring rules in the case of $n = 2$ outcomes. In Figure~\ref{fig:quad_log}, for both of these scoring rules, we show the difference between the expert's reward if a given outcome happens and their reward if it does not happen, as a function of the probability that they assign to the outcome.\footnote{We scale down the logarithmic rule by a factor of $2 \ln 2$ to make the two rules comparable. The factor $2 \ln 2$ was chosen to make the range of values taken on by the Savage representations of the two scoring rules the same (see Section~\ref{subsec:definitions}).}

\begin{figure}[ht]
	\centering
	\includegraphics[scale=1]{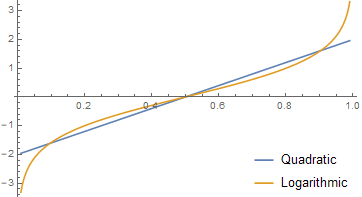}
	\caption{Difference between expert's reward if an outcome happens and if it does not happen, as a function of the expert's report, for the quadratic and logarithmic scoring rules. For example, if the expert reports a 70\% probability of an outcome, then under the quadratic rule they receive a score of $2 \cdot 0.7 - 0.7^2 - 0.3^2 = 0.82$ if the outcome happens and $2 \cdot 0.3 - 0.7^2 - 0.3^2 = 0.02$ if it does not: a difference of $0.8$. If rewarded with the logarithmic rule, this difference would be $0.61$.}
	\label{fig:quad_log}
\end{figure}

This difference scales linearly with the expert's report for the quadratic rule. Meanwhile, for the logarithmic rule, the difference changes more slowly than for the quadratic rule for probabilities in the middle, but much more quickly at the extremes. Informally speaking, this means that the logarithmic rule indicates a preference (of the elicitor) for high precision close to $0$ and $1$, while the quadratic rule indicates a more even preference for precision across $[0, 1]$. Put another way, an elicitor who chooses to use the logarithmic scoring rule renders a judgment that the probabilities $0.01$ and $0.001$ are quite different; one who uses the quadratic rule indicates that these probabilities are very similar.\\

On its surface, the \emph{elicitation} of forecasts has seemingly little to do with their \emph{aggregation}. However, given that the choice of scoring rule implies a subjective judgment about how different probabilities compare to one another, it makes sense to apply this judgment to the aggregation of forecasts as well.

As an example, consider the setting of weather prediction, with models playing the role of experts. In such contexts we often care about low-probability extreme events: a 0.1\% chance of an imminent major hurricane may not be worth preparing for; a 1\% chance could mean significant preparations, and a 10\% chance could mean mandatory evacuations. The need to distinguish very unlikely events from somewhat unlikely events has two consequences. First, as discussed above, this is a reason to use the logarithmic scoring rule to assess the quality of weather models. Second, we wish to avoid the failure mode in which an ill-informed forecaster assigns a high probability due to lack of evidence and thereby drowns out a better-informed low-probability forecast. We would expect to encounter this failure mode with linear pooling; for example, if the more informed model predicts a 0.1\% chance and the less informed model predicts a 20\% chance, linear pooling with equal weights\footnote{Assigning equal weights make sense when there is not enough information to predict in advance which model will be less informed.} would predict roughly a 10\% chance. Logarithmic pooling, by contrast, assigns roughly a 1.6\% chance to the event, avoiding this failure mode. In general, a calibrated model that predicts a very low probability must have good evidence, so it may make sense to give the model more weight.

By contrast, a consultant whose job is to determine the closest races in an election may not much care about the difference between a 0.1\% and a 1\% chance of victory. After all, attention and resources are generally devoted to races with highly uncertain outcomes. As such, it might make sense to assess the qualities of forecast models using the quadratic scoring rule. Similarly, without a compelling reason to pay extra attention to extreme probabilities, it may make more sense to simply take the average of the forecasts' opinions.

In the case of extreme weather prediction, we have argued in favor of using the logarithmic scoring rule to assess the models and the logarithmic pooling method to aggregate them, for similar reasons. In the case of political prediction for targeting close races, we have argued in favor of the quadratic scoring rule and linear pooling, also for similar reasons. Could there be a formal connection between proper scoring rules and opinion pooling methods that captures this intuition? This brings us to the main focus of our paper: namely, we prove a novel correspondence between proper scoring rules and opinion pooling methods.

\subsection{Our definitions} \label{subsec:definitions}
Before introducing the aforementioned correspondence, we need to introduce the \emph{Savage representation} of a proper scoring rule.

\paragraph{\textbf{Savage representation}} A proper scoring rule has a unique representation in terms of its expected reward function $G$, i.e. the expected score on an expert who believes (and reports) a distribution $\vect{p}$:
\[G(\vect{p}) := \EE[j \leftarrow \vect{p}]{s(\vect{p}; j)} = \sum_{j = 1}^n p(j) s(\vect{p}; j).\]
This representation of $s$, introduced in \cite{sav71}, is known as the \emph{Savage representation}, though we will usually refer to it as the \emph{expected reward function}. Given that $s$ is proper, $G$ is strictly convex; and conversely, given a strictly convex function $G$, one can re-derive $s$ with the formula
\begin{equation} \label{eq:s_from_g_1}
s(\vect{p}; j) = G(\vect{p}) + \angles{\vect{g}(\vect{p}), \delta_j - \vect{p}},
\end{equation}
where $\vect{g}$ is the gradient\footnote{Or a subgradient, if $G$ is not differentiable.} of $G$ \cite{gr07}. Pictorially, draw the tangent plane to $G$ at $\vect{p}$; then the expert's score if outcome $j$ is realized is the height of the plane at $\delta_j$.

The Savage representation of the quadratic scoring rule is $G_{\text{quad}}(\vect{p}) = \sum_{j = 1}^n p(j)^2$. The Savage representation of the logarithmic scoring rule is $G_{\log}(\vect{p}) = \sum_{j = 1}^n p(j) \ln p(j)$.

The function $\vect{g}$, which will be central to our paper, describes the difference in the expert's score depending on which outcome happens. More precisely, the vector $(s(\vect{p}; j_1), \dots, s(\vect{p}; j_n))$ is exactly the vector $\vect{g}(\vect{p})$, except possibly for a uniform translation in all coordinates. For example, $s(\vect{p}; j_1) - s(\vect{p}; j_2) = g_1(\vect{p}) - g_2(\vect{p})$; this is precisely the quantity plotted in Figure~\ref{fig:quad_log} for the quadratic and logarithmic scoring rules. This observation about the function $\vect{g}$ motivates the connection that we will establish between proper scoring rules and opinion pooling methods.

\paragraph{\textbf{Quasi-arithmetic opinion pooling}}
We can now define our correspondence between proper scoring rules and opinion pooling methods. Given a proper scoring rule $s$ used for elicitation, and given $m$ probability distributions $\vect{p}_1, \dots, \vect{p}_m$ and expert weights $w_1, \dots, w_m$, the aggregate distribution $\vect{p}^*$ that we suggest is the one satisfying
\[\vect{g}(\vect{p}^*) = \sum_{i = 1}^m w_i \vect{g}(\vect{p}_i).\]
(In Section~\ref{sec:prelims} we will define this notion more precisely using subgradients of $G$ instead of gradients; this will ensure that $\vect{p}^*$ is well defined, i.e. that it exists and is unique.) This definition of $\vect{p}^*$ can be restated as the forecast that minimizes the weighted average Bregman divergence (with respect to $G$) to all experts' forecasts.

We refer to this pooling method as \emph{quasi-arithmetic pooling} with respect to $\vect{g}$ (or the scoring rule $s$), or \emph{QA pooling} for short.\footnote{This term comes from the notion of quasi-arithmetic means: given a continuous, strictly increasing function $f$ and values $x_1, \dots, x_m$ , the \emph{quasi-arithmetic mean with respect to $f$} of these values is $f^{-1}(1/m \sum_i f(x_i))$.} To get a sense of QA pooling, let us determine what this method looks like for the quadratic and logarithmic scoring rules.

\paragraph{\textbf{QA pooling with respect to the quadratic scoring rule}} We have $\vect{g}_{\text{quad}}(\vect{x}) = (2x_1, \dots, 2x_n)$, so we are looking for the $\vect{p}^*$ such that
\[(2p^*(1), \dots, 2p^*(n)) = \sum_{i = 1}^m w_i (2p_i(1), \dots, 2p_i(n)).\]
This is $\vect{p}^* = \sum_{i = 1}^m w_i \vect{p}_i$. Therefore, \emph{QA pooling for the quadratic scoring rule is precisely linear pooling}.

\paragraph{\textbf{QA pooling with respect to the logarithmic scoring rule}}
We have $\vect{g}_{\log}(\vect{x}) = (\ln x_1 + 1, \dots, \ln x_n + 1)$, so we are looking for the $\vect{p}^*$ such that
\[(\ln p^*(1) + 1, \dots, \ln p^*(n) + 1) = \sum_{i = 1}^m w_i (\ln p_i(1) + 1, \dots, \ln p_i(n) + 1).\]
By exponentiating the components on both sides, we find that $p^*(j) = c\prod_{i = 1}^n (p_i(j))^{w_i}$ for all $j$, for some proportionality constant $c$. \emph{This is precisely the definition of the logarithmic pooling method.} (The constant $c$ comes from the fact that values of $\vect{g}(\cdot)$ should be interpreted modulo translation by the all-ones vector; see Remark~\ref{remark:mod_T1n}.)\\

The fact that this pooling scheme maps the two most well-studied scoring rules to the two most well-studied opinion pooling methods has not been noted previously, to our knowledge. This correspondence suggests that --- beyond just our earlier informal justification --- QA pooling with respect to a given scoring rule may be a fundamental concept. The rest of this paper argues that this is indeed the case.

This correspondence may have practical implications for forecasters. While the quadratic and logarithmic scoring rules are both ubiquitous in practice, linear pooling is far more common than logarithmic pooling \cite{rg10}. This is despite empirical evidence that logarithmic pooling often outperforms linear pooling \cite{sat14}. The connection that we establish between the logarithmic scoring rule and logarithmic pooling provides further reason to think that logarithmic pooling has been somewhat overlooked.

\subsection{Our results}
\paragraph{\textbf{(Section~\ref{sec:max_min}) Max-min optimality}} Suppose that a principal asks you to issue a forecast and will pay you according to $s$. You are not knowledgeable on the subject but know some experts whom you trust on the matter (perhaps to varying degrees). You sub-contract the experts, promising to pay each expert $i$ according to $w_i \cdot s$. By using QA pooling according to $s$ on the experts' forecasts, you guarantee yourself a profit; in fact, this strategy maximizes your worst-case profit, and is the unique such report. Furthermore, this profit is the same for all outcomes. This fact can be interpreted to mean that you have, in a sense, pooled the forecasts ``correctly": you do not care which outcome will come to pass, which means that you have correctly factored the expert opinions into your forecast. We give an additional interpretation of this optimality notion as maximizing an aggregator's guaranteed improvement over choosing an expert at random.

\paragraph{\textbf{(Section~\ref{sec:convex_losses}) Learning expert weights}} Opinion pooling entails assigning weights to experts. Where do these weights come from? How might one learn them from experience?

Suppose we have a fixed proper scoring rule $s$, and further consider fixing the reports of the $m$ experts as well as the eventual outcome. One can ask: what does the score of the aggregate distribution (per QA pooling with respect to $s$) look like as a function of $\vect{w}$, the vector of expert weights? We prove that this function is concave. This is useful because it allows for online convex optimization over expert weights.

\begin{nonothm}[informal]
Let $s$ be a bounded proper scoring rule.\footnote{For which QA pooling is well defined (we discuss this below).} For time steps $t = 1 \dots T$, $m$ experts report forecasts to an aggregator, who combines them into a forecast $\vect{p}^t$ using QA pooling with respect to $s$ and suffers a loss of $-s(\vect{p}^t; j^t)$, where $j^t$ is the outcome at time step $t$. If the aggregator updates the experts' weights using online gradient descent, then the aggregator's regret compared to the best weights in hindsight is $O(\sqrt{T})$.
\end{nonothm}

The aforementioned concavity property is a nontrivial fact that demonstrates an advantage of QA pooling over e.g. linear and logarithmic pooling: these pooling methods satisfy the concavity property for some proper scoring rules $s$ but not others.

\paragraph{\textbf{(Section~\ref{sec:axiomatization}) Natural axiomatization for QA pooling methods}} \cite{kol30} and \cite{nag30} independently came up with a simple axiomatization of quasi-arithmetic means. We show how to change these axioms to allow for weighted means; the resulting axiomatization is a natural characterization of all quasi-arithmetic pooling methods in the case of $n = 2$ outcomes. Furthermore, although quasi-arithmetic means are typically defined for scalar-valued functions, we demonstrate that these axioms can be extended to describe quasi-arithmetic means with respect to vector-valued functions, as is necessary for our purposes if $n > 2$. This extension is nontrivial but natural, and to our knowledge has not previously been described.\footnote{For $n > 2$, these axioms characterize the class of all QA pooling methods with respect proper scoring rules that satisfy \emph{convex exposure}, a natural condition that we introduce in Section~\ref{sec:prelims}.}

%% file: 2_Related_Work.tex
\section{Related work} \label{sec:related_work}
\paragraph{\textbf{Opinion pooling}} \cite{cw07} categorize mathematical approaches to opinion pooling as either \emph{Bayesian} or \emph{axiomatic}. A Bayesian approach to this problem is one that entails Bayes updating on each expert's opinion. While quite natural, Bayesian opinion pooling is difficult to apply and, in full generality, computationally intractable. This is because the Bayes updates must fully account for interdependencies between expert opinions.

By contrast, axiomatic approaches do not make assumptions about the structure of information underlying the experts' opinions; instead, they aim to come up with pooling methods that satisfy certain axioms or desirable properties. Such axioms include unanimity preservation, eventwise independence, and external Bayesianality; see e.g. \cite{dl14} for statements of these axioms. For $n > 2$ outcomes, linear pooling is the only method that is both unanimity preserving and eventwise independent \cite{aw80}; however, it is not externally Bayesian. On the other hand, logarithmic pooling is both unanimity preserving and externally Bayesian \cite{gen84}. It is not the only such method, but it is arguably the most natural. Our approach has an axiomatic flavor (though it differs substantially from previous axiomatic approaches).

The earliest comprehensive treatment of opinion pooling in a mathematical setting was by Genest and Zidek \cite{gz86}. See \cite{cw07} and \cite{dl14} for more recent high-level overviews of the subject, including discussion of both Bayesian and axiomatic approaches, as well as discussion of the axioms we discussed in Section~\ref{subsec:pooling}. Recent work on Bayesian opinion pooling includes \cite{cl13}, \cite{sbfmtu14}, \cite{fck15}, and \cite{abs18}. Finally, see \cite{fes20} for work on learning weights for linear pooling under a quadratic loss function; this is an instance of our more general theory in Section~\ref{sec:convex_losses}. For a more detailed discussion of axiomatic approaches, see \cite{acr12}. See also \cite{abbas09} for a non-Bayesian approach that is closely related to ours; our work generalizes \cite[Proposition 4]{abbas09}.

Part of our work provides an alternative interpretation of prior work on pooling via minimizing Bregman divergence, see e.g. \cite{ada14} and \cite{pet19}. (We define Bregman divergence in Section~\ref{sec:prelims}.) Concretely, \cite[\S4]{pet19} defines a notion of pooling analogous to ours, though in a different context. The main focus of their line of work is on connecting opinion pooling to Bregman divergence; our approach connects opinion pooling to proper scoring rules, and a connection to Bregman divergence falls naturally out of this pursuit.

\paragraph{\textbf{Scoring rules}} The literature on scoring rules is quite large; we recommend \cite{gr07} for a thorough but technical overview, or \cite{car16} for a less technical overview that focuses more on applications (while still introducing the basic theory). Seminal work on the theory behind scoring rules includes Brier's paper introducing the quadratic rule \cite{brier50}, Good's paper introducing the logarithmic rule \cite{good52}, and Savage's work on the general theory of proper scoring rules \cite{sav71}. Additionally, see \cite{dm14} for an overview of various families of proper scoring rules.

\paragraph{\textbf{Max-min optimality}} Gr\"{u}nwald and Dawid establish a connection between minimizing worst-case expected loss and maximizing a generalized notion of entropy \cite{gd04}. The application of their quite general work to the domain of proper scoring rules states that a forecaster who seeks to maximize their worst-case score (over outcomes) ought to report the minimizer of the expected reward function of the scoring rule. Although they do not discuss opinion pooling, their work is closely related to our Theorem~\ref{thm:max_min}; we discuss the connection in more detail in Section~\ref{sec:max_min}.

\paragraph{\textbf{Dual averaging}} One perspective on QA pooling is that, instead of directly averaging experts' forecasts, QA pooling prescribes considering forecasts as elements in the dual space of gradients (of the function $G$) and taking the average in this space before converting the result back to the primal space of probabilities. Gradient methods in online machine learning often take the \emph{sum} of gradients of losses. Taking the average is of gradients is a less ubiquitous technique known as \emph{dual averaging}, which was introduced by Nesterov \cite{nesterov09} and generalized further by Xiao \cite{xiao10}. The contexts or QA pooling and dual averaging are quite different, and interpretations of QA pooling in the context of dual averaging appear to be fairly unnatural.

\paragraph{\textbf{Aggregation via prediction markets}} One common way to aggregate probabilistic forecasts is through prediction markets, some of which are based on scoring rules. \cite{hanson03} introduced \emph{market scoring rules} (MSRs), in which experts are sequentially presented with an opportunity to update an aggregate forecast and are rewarded (or penalized) by the amount that their update changed the aggregate prediction's eventual score. \cite{cp07} introduced \emph{cost-function markets}, in which a market maker sells $n$ types of shares --- one for each outcome --- where the price of a share depends on the number of shares sold thus far according to some cost function. They established a connection between cost-function markets and MSRs, where a market with a given cost function will behave the same way as a certain MSR. In particular, the cost function $C$ of a cost function market is the convex dual of the expected reward function $G$ of the proper scoring rule associated with the corresponding MSR \cite[\S8.3]{acv13}.

Subsequent work explored this area further, tying cost-function market making to online learning of probability distributions \cite{cv10} \cite{acv13}. This work differs from ours in that the goal of their online learning problem is to learn a \emph{probability distribution over outcomes}, whereas our goal in Section~\ref{sec:convex_losses} is to learn \emph{expert weights}.

QA pooling has a simple interpretation in terms of cost function markets: for the cost function market corresponding to the scoring rule $s$, let $\vect{q}_i$ be the quantity vector that implies each expert $i$'s probability $\vect{p}_i$ (or in other terms, $i$ would buy a bundle $\vect{q}_i$ of shares as the first participant in the market). Then the QA pool with respect to $s$ is the probability implied by the weighted average quantity vector $\sum_i w_i \vect{q}_i$.

This interpretation stands in contrast to most past work on MSRs and cost function markets. Typically experts trade in series rather than in parallel, and incentives are set up so that an expert brings the market into alignment with their own opinion, rather than an aggregate. Thus, in the well-studied setting of traders whose beliefs do not depend on previous traders' actions, the final state of such a market reflects only the beliefs of the most recent trader, rather than an aggregate of beliefs. One exception to this paradigm is \cite{hlpv18}, which studies traders who trade in series, although in a completely different context.

\paragraph{\textbf{Arbitrage from collusion}} Part of our work can be viewed as a generalization of previous work by Chun and Shachter done in a different context: namely, preventing colluding experts from exploiting arbitrage opportunities \cite{cs11}. The authors show that for the case of $n = 2$ outcomes, if experts are rewarded with the same scoring rule $s$, preventing this is impossible: the experts can successfully collude by all reporting what we are calling the QA pool of their reports with respect to $s$. Our Theorem~\ref{thm:max_min} recovers this result as a special case. See \cite{cdpv14} for related work in the context of wagering mechanisms and \cite{fppw20} for follow-up work on preventing arbitrage from colluding experts.

\paragraph{\textbf{Prediction with expert advice}} In Section~\ref{sec:convex_losses} we discuss learning expert weights online. The online learning literature is vast, but our approach fits into the framework of \emph{prediction with expert advice}. In this setting, at each time step each expert submits a report (in our context a probability distribution). The agent then submits a report based on the experts' submissions, and suffers a loss depending on this report and the eventual outcome. See \cite{cl06} for a detailed account of this setting; the authors prove a variety of no-regret bounds in this setting, ranging (depending form the setting) from $O(\sqrt{T})$ to $O(1)$. Our setting is an ambitious one: while typically one desires low regret compared to the best expert in hindsight, we desire low regret compared to the best \emph{mixture} of experts in hindsight. In \cite[\S3.3]{cl06}, the authors discuss this more ambitious goal, proving (in our terminology) that for a certain class of losses --- bounded exp-concave losses --- it is possible to achieve $O(\log T)$ regret in comparison with the best linear pool of experts in hindsight. This setting is different from ours in two important ways: first, the losses that we consider are not in general exp-concave (e.g. the quadratic loss); and second, the authors consider linear pooling for any loss, whereas we consider QA pooling with respect to the loss function.

\paragraph{\textbf{Quasi-arithmetic means}} Our notion of quasi-arithmetic pooling is an adaptation (and extension to higher dimensions) of the existing notion of quasi-arithmetic means. These were originally defined and axiomatized independently in \cite{kol30} and \cite{nag30}. Acz\'{e}l generalized this work to include weighted quasi-arithmetic means \cite{aczel48}, though these means have weights baked in rather than taking them as inputs, which is different from our setting. See \cite[\S3.1]{gmmp11} for an overview of this topic.

%% file: 3_Preliminaries.tex
\section{Preliminaries} \label{sec:prelims}
Throughout this paper, we will let $m$ be the number of experts and use the index $i$ to refer to any particular expert. We will let $n$ be the number of outcomes and use the index $j$ to refer to any particular outcome. We will also let $\Delta^n$ be the standard simplex in $\RR^n$, i.e.\ the one with vertices $\delta_1, \dots, \delta_n$. (Here, $\delta_j$ denotes the vector with a $1$ in the $j$-th coordinate and zeros elsewhere.)

\subsection{Proper scoring rules}
A \emph{scoring rule} is a function $s: \Delta^n \times [n] \to \RR \cup \{-\infty\}$. The interpretation of $s$ is that an expert receives reward $s(\vect{p}; j)$ if they report $\vect{p}$ and outcome $j$ happens. A scoring rule $s$ is \emph{proper} if for all $\vect{p} \in \Delta^n$, we have
\[\EE[j \leftarrow \vect{p}]{s(\vect{p}; j)} \ge \EE[j \leftarrow \vect{p}]{s(\vect{x}; j)}\]
for all $\vect{x} \in \Delta^n$, with equality only when $\vect{x} = \vect{p}$. (Here, $j \leftarrow \vect{p}$ means that $j$ is drawn randomly from the probability distribution $\vect{p}$.)\footnote{Some authors refer to such scoring rules as \emph{strictly} proper while others define ``proper" to entail strictness; we choose the latter convection.} We will henceforth assume that $s$ is \emph{regular}, meaning that $s(\cdot; j)$ is real-valued for all $j$, except possibly that $s(\vect{p}; j) = -\infty$ if $p_j = 0$ \cite[Definition 2]{gr07}.

We define the \emph{forecast domain} $\mathcal{D}$ associated with $s$ to be the set of forecasts $\vect{p}$ such that $s(\vect{p}; j)$ is real-valued for all $j$. When discussing forecast aggregation, we will assume that all forecasts belong to $\mathcal{D}$.\footnote{This choice removes from consideration cases such as two experts reporting $(1, 0)$ and $(0, 1)$ under the logarithmic scoring rule; aggregating these forecasts using our method is tantamount to adding positive and negative infinity.}

Given a proper scoring rule $s$, we define its \emph{expected reward function} $G: \Delta^n \to \RR$ by
\[G(\vect{p}) := \EE[j \leftarrow \vect{p}]{s(\vect{p}; j)} = \sum_{j = 1}^n p(j) s(\vect{p}; j).\]

We will frequently be relying on an alternative representation of $s$ -- sometimes known as the \emph{Savage representation} -- in terms of $G$.

\begin{prop}[{\cite[Theorem 2]{gr07}}] \label{prop:G_convex}
	A regular scoring rule $s$ is proper if and only if
	\begin{equation} \label{eq:s_from_g_2}
		s(\vect{p}; j) = G(\vect{p}) + \angles{\vect{g}(\vect{p}), \delta_j - \vect{p}},
	\end{equation}
	for some convex function $G: \Delta^n \to \RR$ and subgradient function\footnote{That is, $\vect{g}$ satisfies $G(\vect{y}) \ge G(\vect{x}) + \angles{\vect{g}(\vect{x}), \vect{y} - \vect{x}}$ for all $\vect{x}, \vect{y}$. Note that $g_i(\vect{x})$ may be $-\infty$ if $x_i = 0$; see \cite{waggoner21} for an examination of subgradients of functions to the extended reals.} $\vect{g}$ of $G$. The function $G$ is then the expected reward function of $s$.
\end{prop}

We will henceforth assume that $s(\cdot; j)$ is continuous on $\Delta^n$ for all $j$, as is $G(\cdot)$; to our knowledge, this is the case for all frequently-used proper scoring rules.\footnote{For $s$, continuity is with respect to the standard topology on $\RR \cup \{-\infty\}$, i.e.\ the one that includes sets of the form $[-\infty, r)$ as open sets.} By Equation~\ref{eq:s_from_g_2}, this means that $\vect{g}$ is also continuous on $\Delta^n$.\footnote{Continuity of each component $g_\ell$ of $\vect{g}$ is (as with $s$) with respect to the standard topology on $\RR \cup \{-\infty\}$. To see that $g_\ell$ is continuous at a given $\vect{p}$, consider the limit of Equation~\ref{eq:s_from_g_2} as $\vect{x} \to \vect{p}$ with $j = \ell$ if $p_j = 0$ and $j \neq \ell$ if $p_j \neq 0$.} A convex function with a continuous finite subgradient is differentiable \cite[Proposition 17.41]{bc11}, which means that $G$ is differentiable on the interior of $\Delta^n$, with gradient $\vect{g}$.

The important intuition to keep in mind for Equation~\ref{eq:s_from_g_2} is that the score of an expert who reports $\vect{p}$ is determined by drawing the tangent plane to $G$ at $\vect{p}$; the value of this plane at $\delta_j$, where $j$ is the outcome that happens, is the expert's score.

We refer to $\vect{g}$ as the \emph{exposure function} of $s$. We borrow this term from finance, where \emph{exposure} refers to how much an agent stands to gain or lose from various possible outcomes --- informally speaking, how much the agent cares about which outcome will happen. If we view $G(\vect{p}) - \angles{\vect{g}(\vect{p}), \vect{p}}$ as the agent's ``baseline profit," then the $j$-th component of $\vect{g}(\vect{p})$ is the amount that the agent stands to gain (or lose) on top of the baseline profit if outcome $j$ happens.\\

Following e.g.\ \cite{gr07}, we give a geometric intuition for Proposition~\ref{prop:G_convex} to help explain why the \emph{properness} of $s$ corresponds to the \emph{convexity} of $G$; this intuition will be helpful for understanding Bregman divergence (below) and the proofs in Section~\ref{sec:max_min}.

Consider Figure~\ref{fig:quad_example}, which depicts some $G$ in the $n = 2$ outcome case, with the $x$-axis corresponding to the probability of Outcome 1 (see Remark~\ref{remark:n2} for formal details). Suppose that the expert \emph{believes} that the probability Outcome 1 is $0.7$. If the expert \emph{reports} $p = 0.7$, then their reward in the cases of Outcome 1 and Outcome 2 are the $y$-values of the rightmost and leftmost points on the red line, respectively. Thus, in expectation, the expert's reward is the $y$-value of the red point. If instead the expert lies and reports $p = 0.4$ (say), then the $y$-values of the rightmost and leftmost points on the blue line represent the expert's rewards in the cases of Outcome 1 and Outcome 2, respectively. In this case, since Outcome 1 is still 70\% likely, the expert's expected reward is the $y$-value of the blue point. \emph{Because $G$ is strictly convex}, the blue point is strictly below the red point; that is, the expert is strictly better off reporting $p = 0.7$. This argument holds in full generality: for any strictly convex function $G$ in any number of dimensions.\\

\begin{figure}[ht]
	\centering
	\includegraphics[scale=0.7]{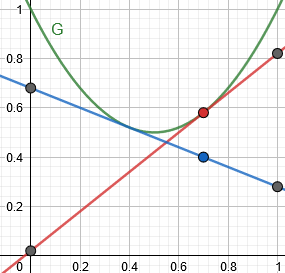}
	\caption{A convex function $G$, with tangent lines drawn at $x = 0.7$ and $x = 0.4$. If an expert believes that the probability of an event is $0.7$, their expected score if they report $p = 0.7$ is the $y$-value of the red point; if they instead report $p = 0.4$, their expected score is the $y$-value of the blue point. Because $G$ is convex, the red point is guaranteed to be above the blue point, so the expert is incentivized to be truthful.}
	\label{fig:quad_example}
\end{figure}

We will find \emph{Bregman divergence} to be a useful concept for some of our proofs.
\begin{defin}[Bregman divergence]
Given a differentiable, strictly convex function $G: S \to \RR$ with gradient $\vect{g}$, where $S$ is a convex subset of $\RR^n$, and given $\vect{p}, \vect{q} \in S$, the \emph{Bregman divergence} between $\vect{p}$ and $\vect{q}$ \emph{with respect to $G$} is
\[D_G(\vect{p} \parallel \vect{q}) := G(\vect{p}) - G(\vect{q}) - \angles{\vect{g}(\vect{q}), \vect{p} - \vect{q}}.\]
\end{defin}

(Note that Bregman divergence is not symmetric.) A geometric interpretation of $D_G(\vect{p} \parallel \vect{q})$ is: if you draw the tangent plane to $G$ at $\vect{q}$, how far below $G(\vect{p})$ the value of that plane will be at $\vect{p}$. For example, the distance between the red and blue points in Figure~\ref{fig:quad_example} is the Bregman divergence between the expert's belief $(0.7, 0.3)$ and their report $(0.4, 0.6)$. This gives the following interpretation of Bregman divergence in the context of proper scoring rules:

\begin{remark} \label{remark:bregman_scoring}
If $G$ is the expected reward function of a proper scoring rule, then $D_G(\vect{p} \parallel \vect{q})$ is the expected reward lost by reporting $\vect{q}$ when your belief is $\vect{p}$. Put otherwise, $D_G(\vect{p} \parallel \vect{q})$ measures the ``wrongness" of the report $\vect{q}$ relative to a correct answer of $\vect{p}$.
\end{remark}

\begin{prop}[Well-known facts about Bregman divergence] \leavevmode \label{prop:bregman_facts}
\begin{itemize}
\item $D_G(\vect{p} \parallel \vect{q}) \ge 0$, with equality only when $\vect{p} = \vect{q}$.
\item For any $\vect{q}$, $D_G(\vect{x} \parallel \vect{q})$ is a strictly convex function of $\vect{x}$.
\end{itemize}
\end{prop}

Finally, we make a note about interpreting the $n = 2$ outcome case in one dimension.

\begin{remark} \label{remark:n2}
Because $\Delta^n$ is $(n - 1)$-dimensional, we can think of the case of $n = 2$ outcomes in one dimension. All probabilities in are of the form $(p, 1 - p)$; we map $\Delta^2$ to $[0, 1]$ via the first coordinate. Thus, we let $G(p) := G(p, 1 - p)$. We let $g(p) := G'(p) = \angles{\vect{g}(p, 1 - p), (1, -1)}$. The tangent line to $G$ at $p$, e.g. as in Figure~\ref{fig:quad_example}, will intersect the line $x = 1$ at $s(p; 1)$ (i.e. the score if Outcome 1 happens) and intersect the line $x = 0$ at $s(p; 2)$ (i.e. the score if Outcome 2 happens). This formulation will be helpful when discussing the two-outcome case, e.g. in Section~\ref{sec:axiomatization}.
\end{remark}

\subsection{Quasi-arithmetic pooling}
We now introduce the central concept of this paper, quasi-arithmetic pooling.

\begin{defin}[quasi-arithmetic pooling] \label{def:qa_pooling}
	Let $s$ be a proper scoring rule with expected reward function $G$ and exposure function $\vect{g}$. Given forecasts $\vect{p}_1, \dots, \vect{p}_m \in \mathcal{D}$ with non-negative weights $w_1, \dots, w_m$ adding to $1$, the \emph{quasi-arithmetic (QA) pool} of these forecasts \emph{with respect to $s$} (or with respect to $\vect{g}$), denoted by $\sideset{}{_\vect{g}}\bigoplus\limits_{i = 1}^m (\vect{p}_i, w_i)$, is the unique $\vect{p}^* \in \Delta^n$ such that $\sum_{i = 1}^m w_i \vect{g}(\vect{p}_i)$ is a subgradient of $G$ at $\vect{p}^*$.
	
	If the forecasts and weights are clear from context, we may simply write $\vect{p}^*$ to refer to their quasi-arithmetic pool; or, if only the forecasts are clear, we may write $\vect{p}^*_{\vect{w}}$, where $\vect{w}$ is the vector of weights.
\end{defin}

\begin{remark}
Equivalently \cite[Theorem 23.5]{roc70_book}, we can define
\begin{equation} \label{eq:qa_min_1}
	\sideset{}{_\vect{g}}\bigoplus\limits_{i = 1}^m (\vect{p}_i, w_i) := \arg \min_{\vect{x}} G(\vect{x}) - \angles{\vect{x}, \sum_{i = 1}^m w_i \vect{g}(\vect{p}_i)}.
\end{equation}
Also equivalently, we can define
\[\sideset{}{_\vect{g}}\bigoplus\limits_{i = 1}^m (\vect{p}_i, w_i) := \arg \min_{\vect{x}} \sum_{i = 1}^m w_i D_G(\vect{x} \parallel \vect{p}_i).\]
This expression differs from the one in Equation~\ref{eq:qa_min_1} by a constant: namely, $\sum_i w_i(\angles{\vect{p}_i, \vect{g}(\vect{p}_i)} - G(\vect{p}_i))$.
\end{remark}

The QA-pool is well-defined, i.e.\ exists and is unique. It is unique because a strictly convex function cannot have the same subgradient at two different points \cite[Lemma 3.11]{waggoner21}. It exists because $G(\vect{x}) - \angles{\vect{x}, \sum_i w_i \vect{g}(\vect{p}_i)}$ is continuous and thus attains its minimum on the (compact) domain $\Delta^n$.

In light of the fact that $D_G(\vect{p} \parallel \vect{q})$ is the expected reward lost by an expert who, believing $\vect{p}$, reports $\vect{q}$ (see Remark~\ref{remark:bregman_scoring}), the Bregman divergence formulation of QA pooling gives another natural interpretation.

\begin{remark} \label{remark:bregman_interpretation}
	Consider a proper scoring rule $s$ with forecasts $\vect{p}_1, \dots, \vect{p}_m$ with weights $w_1, \dots, w_m$. The QA pool of these forecasts is the forecast $\vect{p}^*$ that, if it is the correct answer (i.e. if the outcome is drawn according to $\vect{p}^*$), would minimize the expected loss of a randomly chosen (according to $\vect{w}$) expert relative to reporting $\vect{p}^*$.
\end{remark}

In this sense, QA pooling reflects a \emph{compromise} between experts: it is the probability that, if it were correct, would make the experts' forecast least wrong overall.

\begin{remark}
	Since the Bregman divergence is convex in its first argument, computing the QA pool is a matter of convex optimization. In particular, given oracle access to $\vect{g}$, the ellipsoid method can be used to efficiently find the QA pool of a list of forecasts.
\end{remark}

Note that although Definition~\ref{def:qa_pooling} only specifies that $\vect{p}^* \in \Delta^n$, in fact it lies in $\mathcal{D}$:

\begin{claim}
	For any $\vect{p}_1, \dots, \vect{p}_m$ and $w_1, \dots, w_m$, the QA pool $\vect{p^*}$ lies in $\mathcal{D}$.
\end{claim}

\begin{proof}
	Let $\vect{x} \in \Delta^n \setminus \mathcal{D}$. We show that $\sum_i w_i \vect{g}(\vect{p}_i)$ is not a subgradient of $G$ at $\vect{x}$.
	
	We have $s(\vect{x}; j) = -\infty$ for some $j$ (satisfying $x_j = 0$), so $\angles{\vect{g}(\vect{x}), \delta_j - \vect{x}} = -\infty$. Since $\vect{g}$ is continuous, for sufficiently small $\epsilon$, we have
	\[\angles{\vect{g}(\vect{x} + \epsilon'(\delta_j - \vect{x})), \delta_j - \vect{x}} < \angles{\sum_i w_i \vect{g}(\vect{p}_i), \delta_j - \vect{x}}\]
	for all $\epsilon' \le \epsilon$. This means that
	\[G(\vect{x} + \epsilon(\delta_j - \vect{x})) - G(\vect{x}) < \angles{\sum_i w_i \vect{g}(\vect{p}_i), \epsilon(\delta_j - \vect{x})},\]
	so $\sum_i w_i \vect{g}(\vect{p}_i)$ is not a subgradient of $G$ at $\vect{x}$, as desired.
\end{proof}

While our max-min optimality result (Section~\ref{sec:max_min}) holds unconditionally, our results in Section~\ref{sec:convex_losses} and \ref{sec:axiomatization} require that our proper scoring rule $s$ satisfy a property that we term \emph{convex exposure}.

\begin{defin}[convex exposure] \label{def:convex_exposure}
	A proper scoring rule $s$ with forecast domain $\mathcal{D}$ has \emph{convex exposure} if the range of its exposure function $\vect{g}$ on $\mathcal{D}$ is a convex set.
\end{defin}

The key fact about proper scoring rules with convex exposure is that for all $\vect{p}_1, \dots, \vect{p}_m$ and $w_1, \dots, w_m$, $\sum_i w_i \vect{g}(\vect{p}_i) = \vect{g}(\vect{x})$ for some $\vect{x} \in \mathcal{D}$. This means that $\sum_i w_i \vect{g}(\vect{p}_i)$ is a subgradient of $G$ at $\vect{x}$, so $\vect{x}$ is the weighted QA pool of the $p_i$'s. In other words, we may write
\begin{equation} \label{eq:g_convex_exposure}
	\vect{g}(\vect{p}^*) = \sum_i w_i \vect{g}(\vect{p}_i),
\end{equation}
where $\vect{p}^*$ is the weighted QA pool of the given forecasts. The convex exposure property thus allows us to write down relations between exposures of forecasts that would otherwise not necessarily be true.

The quadratic and logarithmic scoring rules, as well as all proper scoring rules for binary outcomes, have convex exposure. In Appendix~\ref{appx:convex_exposure} we explore in more depth the question of which commonly used proper scoring rules have convex exposure property.

\begin{remark} \label{remark:mod_T1n}
Because the domain of $G$ is a subset of $\Delta^n$ (and thus lies in a plane that is orthogonal to the all-ones vector $\vect{1}_n$), it makes the most sense to think of its gradient function $\vect{g}$ as taking on values in $\RR^n$ modulo translation by the all-ones vector $\vect{1}_n$; we will denote this space by $\RR^n/T(\vect{1}_n)$. Sometimes we find it convenient to treat $G$ as a function of $n$ variables rather than $n - 1$ variables out of convenience, thus artificially extending the domain of $G$ outside of the plane containing $\Delta^n$. The component of the gradient of $G$ that is parallel to $\vect{1}_n$ is not relevant.\footnote{Formally, consider the change of coordinates given by $z_j = x_n - x_j$ for $j \le n - 1$ and $z_n = \sum_j x_j$, so that the domain of $G$ lies in the plane $z_n = 1$. Then for $j \le n - 1$, $\partl{G}{z_j}$ at a given point in the domain of $G$ does not change if $1$ is substituted for $z_n$; only $\partl{G}{z_n}$ changes (to zero). Equivalently in terms of our original coordinates, the change that $\vect{g}$ undergoes when we consider $G$ to be a function only defined on $\mathcal{D}$ instead of $\RR^n$ is precisely a projection of $\vect{g}$ onto $\{\vect{x} \in \RR^n: \sum_i x_i = 0\}$.}
\end{remark}

%% file: 4_Max_Min.tex
\section{QA pooling as a max-min optimal method} \label{sec:max_min}
Our goal is to give a formal justification for quasi-arithmetic pooling. Remark~\ref{remark:bregman_interpretation} established that the QA pool is optimizes (i.e.\ minimizes) the weighted average Bregman divergence to the experts' forecasts. This section gives another justification for QA pooling in terms of max-min optimality. We will give additional justifications in later sections.

\begin{theorem} \label{thm:max_min}
Let $s$ be a proper scoring rule and let $\vect{g}$ be the exposure function of $s$. Fix any forecasts $\vect{p}_1, \dots, \vect{p}_m \in \mathcal{D}$ with non-negative weights $w_1, \dots, w_m$ adding to $1$. Define
\[u(\vect{p}; j) := s(\vect{p}; j) - \sum_{i = 1}^m w_i s(\vect{p}_i; j).\]
Then the quantity $\min_j u(\vect{p}; j)$ is uniquely maximized by setting $\vect{p}$ to $\vect{p}^* := \sideset{}{_\vect{g}}\bigoplus\limits_{i = 1}^m (\vect{p}_i, w_i)$. Furthermore, this minimum (across outcomes $j$) is achieved simultaneously by all $j$ with $p^*_j > 0$. This quantity is non-negative, and is positive unless all reports $\vect{p}_i$ with positive weights are equal.
\end{theorem}

One interpretation for this theorem statement is as follows. Consider an agent who is tasked with submitting a forecast, and who will be paid according to $s$. The agent decides to sub-contract $m$ experts to get their opinions, paying expert $i$ the amount $w_i s(\vect{p}_i; j)$ if the expert reports $\vect{p}_i$ and outcome $j$ happens. (Perhaps experts whom the agent trusts more have higher $w_i$'s.) Finally, the agent reports some (any) forecast $\vect{p}$. Then $u(\vect{p}; j)$ is precisely the agent's profit (utility).

The quantity $\min_j u(\vect{p}; j)$ is the agent's minimum possible profit over all outcomes. It is natural to ask which report $\vect{p}$ maximizes this quantity. Theorem~\ref{thm:max_min} states that this maximum is achieved by the QA pool of the experts' forecasts with respect to $s$, and that this is the unique maximizer.

A possible geometric intuition to keep in mind for the proof (below): for each expert $i$, draw the plane tangent to $G$ at $\vect{p}_i$. For any $j$, the value of this plane at $\delta_j$ is $s(\vect{p}_i; j)$. Now take the weighted average of all $m$ planes; this is a new plane whose intersection with any $\delta_j$ is the total reward received by the experts if $j$ happens. Since $G$ is convex, this plane lies below $G$. To figure out which point maximizes the agent's guaranteed profit, push the plane upward until it hits $G$. It will hit $G$ at $\vect{p}^*$ and the agent's worst-case profit will be the vertical distance that the plane was pushed.

\begin{proof}[Proof of Theorem~\ref{thm:max_min}]
By Equation~\ref{eq:qa_min_1}, computing the QA pool amounts to finding the minimizer $\vect{p}^*$ of the function $G(\vect{x}) - \angles{\vect{x}, \sum_{i = 1}^m w_i \vect{g}(\vect{p}_i)}$ over $\Delta^n$. If $\vect{p}^*$ is in the interior of $\Delta^n$, then this expression is differentiable at $\vect{p}^*$. If $\vect{p}^*$ is on the boundary, then the expression can be extended to a differentiable function in a neighborhood around $\vect{p}^*$.\footnote{This follows e.g.\ from \cite[Theorem 1.8]{am15}, where we take $C$ in the theorem statement to be a compact subset of $\Delta^n$ containing $\vect{p}^*$ where $G$ is differentiable. Here we use that $\vect{p}^* \in \mathcal{D}$; by continuity of $\vect{g}$, such a subset necessarily exists.} Thus, applying the KKT conditions \cite[\S5.5.3]{bv04} tells us that
\[\vect{g}(\vect{p}^*) = \sum_{i = 1}^m w_i \vect{g}(\vect{p}_i) + \lambda \vect{1}_n + \pmb{\mu}\]
for some $\lambda \in \RR$ and $\pmb{\mu} \in \RR^n$ such that $\mu_j \ge 0$ and $\mu_j = 0$ if $p^*_j > 0$. We therefore have
\begin{align*}
	u(\vect{p}^*; j) &= s(\vect{p}^*; j) - \sum_i w_i s(\vect{p}_i; j)\\
	&= G(\vect{p}^*) + \angles{\vect{g}(\vect{p}^*), \delta_j - \vect{p}^*} - \sum_i w_i (G(\vect{p}_i) + \angles{\vect{g}(\vect{p}_i), \delta_j - \vect{p}_i})\\
	&= G(\vect{p}^*) - \sum_i w_i G(\vect{p}_i) + \angles{\sum_i w_i \vect{g}(\vect{p}_i) + \lambda \vect{1}_n + \pmb{\mu}, \delta_j - \vect{p}^*} - \sum_i w_i \angles{\vect{g}(\vect{p}_i), \delta_j - \vect{p}_i}\\
	&= G(\vect{p}^*) - \sum_i w_i G(\vect{p}_i) + \sum_i w_i \angles{\vect{g}(\vect{p}_i), \vect{p}_i - \vect{p}^*} + \angles{\lambda \vect{1}_n + \pmb{\mu}, \delta_j - \vect{p}^*}\\
	&= G(\vect{p}^*) - \sum_i w_i G(\vect{p}_i) + \sum_i w_i \angles{\vect{g}(\vect{p}_i), \vect{p}_i - \vect{p}^*} + \mu_j\\
	&= \sum_i w_i D_G(\vect{p}^* \parallel \vect{p}_i) + \mu_j.
\end{align*}
The second-to-last step follows from the fact that $\angles{\lambda \vect{1}_n, \delta_j - \vect{p}^*} = 0$ and $\angles{\pmb{\mu}, \vect{p}^*} = 0$, and the last step follows by the definition of Bregman divergence. Since $\mu_j = 0$ for every $j$ such that $p^*_j > 0$ (of which there is at least one), the minimum of $u(\vect{p}^*; j)$ over $j$ is achieved simultaneously for all $j$ with $p^*_j > 0$, and this minimum is equal to $\sum_i w_i D_G(\vect{p}^* \parallel \vect{p}_i)$. This quantity is non-negative, and positive except when all $\vect{p}_i$'s with positive weights are equal.

Finally we show that $\vect{x} = \vect{p}^*$ maximizes $\min_j u(\vect{x}; j)$. Suppose that for some report $\vect{q}$ we have that $\min_j u(\vect{q}; j) \ge \min_j u(\vect{p}^*; j)$. Then $u(\vect{q}; j) \ge u(\vect{p}^*; j)$ for \emph{every} $j$ such that $p^*_j > 0$. But this means that the expected score (according to $s$) of an expert who believes $\vect{p}^*$ is larger if the expert reports $\vect{q}$ than if the expert reports $\vect{p}^*$. This contradicts the fact that $s$ is proper.
\end{proof}

\begin{remark}
We can reformulate Theorem~\ref{thm:max_min} as follows: suppose that an agent has access to forecasts $\vect{p}_1, \dots, \vect{p}_m$ and needs to issue a forecast, for which the agent will be rewarded using a proper scoring rule $s$. The agent can improve upon selecting an expert at random according to weights $w_1, \dots, w_m$, no matter the outcome $j$, by reporting $\vect{p}^*$. This improvement is the same no matter the outcome -- so long as the outcome is assigned positive probability by $\vect{p}^*$ -- and is a strict improvement unless all forecasts with positive weights are the same.
\end{remark}

\begin{remark}
	Theorem~\ref{thm:max_min} is closely related to work by Gr\"{u}nwald and Dawid that establishes a connection between entropy maximization and worst-case expected loss minimization \cite{gd04}. Their work studies a generalized notion of entropy functions that, for the case of a proper scoring rule $s$, is equal to the negative of the expected reward function $G$. They show that the forecast $\vect{x}$ that maximizes entropy (minimizes $G$) also maximizes an expert's worst-case score (over outcomes $j$). Considering instead the entropy function $\angles{x, \sum_i w_i \vect{g}(\vect{p}_i)} - G(\vect{x})$ yields results that are very similar to our Theorem~\ref{thm:max_min}. In particular, if $\mathcal{D}$ is a closed set or an open set, then the max-min result in Theorem~\ref{thm:max_min} can be derived from \cite[Theorem 5.2]{gd04} and \cite[Theorem 6.2]{gd04}, respectively. However, our proof below captures cases that their results do not address. See also \cite[Lemma 1]{chl19}, from which (upon considering the scoring rule with expected reward function $G(\vect{x}) - \angles{x, \sum_i w_i \vect{g}(\vect{p}_i)}$) yields the simultaneity result of Theorem~\ref{thm:max_min} if $\vect{p}^*$ lies in the interior of $\Delta^n$.
\end{remark}

%% file: 5_Convex_Optimization.tex
\section{Convex losses and learning expert weights} \label{sec:convex_losses}
Thus far, when discussing QA pooling, we have regarded expert weights as given. Where do these weights come from? As demonstrated by the results in this section, if the proper scoring rule $s$ has convex exposure (see Definition~\ref{def:convex_exposure}), these weights can be learned from experience. This learning property for weights falls out of the following key observation, which states that an agent's score is a concave function of the weights it uses for the experts.

\begin{theorem} \label{thm:ws_concave}
Let $s$ be a proper scoring rule with convex exposure and forecast domain $\mathcal{D}$, and fix any $\vect{p}_1, \dots, \vect{p}_m \in \mathcal{D}$. Given a weight vector $\vect{w} = (w_1, \dots, w_m) \in \Delta^m$, define the \emph{weight-score} of $\vect{w}$ for an outcome $j$ as
\[\ws_j(\vect{w}) := s \parens{\sideset{}{_\vect{g}}\bigoplus_{i = 1}^m (\vect{p}_i, w_i); j}.\]
Then for every $j \in [n]$, $\ws_j(\vect{w})$ a concave function of $\vect{w}$.
\end{theorem}

We defer the proof, along with related commentary and results, to Appendix~\ref{appx:convex_losses}. The basic idea of the proof is that for any two weight vectors $\vect{v}$ and $\vect{w}$ and any $c \in [0, 1]$, the quantity
\[\ws_j(c\vect{v} + (1 - c)\vect{w}) - c \ws_j(\vect{v}) - (1 - c) \ws_j(\vect{w})\]
can be expressed as a sum of Bregman divergences, and is therefore non-negative.

\begin{remark}
Theorem~\ref{thm:ws_concave} would \emph{not} hold if the pooling operator in the definition of $\ws$ were replaced by linear pooling or by logarithmic pooling.\footnote{For a counterexample to logarithmic pooling, consider $n = 2$, let $s$ be the quadratic scoring rule, $p_1 = (0.1, 0.9)$, $p_2 = (0.5, 0.5)$, $j = 1$, $\vect{v} = (1, 0)$, $\vect{w} = (0, 1)$, and $c = \frac{1}{2}$. For a counterexample to linear pooling, consider $n = 2$, let $s$ be given by $G(p_1, p_2) = \sqrt{p_1^2 + p_2^2}$ (this is known as the \emph{spherical} scoring rule), $p_1 = (0, 1)$, $p_2 = (0.2, 0.8)$, $j = 1$, $\vect{v} = (1, 0)$, $\vect{w} = (0, 1)$, and $c = \frac{1}{2}$.} This is an advantage of QA pooling over using the linear or logarithmic method irrespective of the scoring rule.
\end{remark}

We now state the no-regret result that we have alluded to. The algorithm referenced in the statement is an application of the standard online gradient descent algorithm (see e.g. \cite[Theorem 3.1]{hazan19}) to our particular setting. We defer both the algorithm and the proof of Theorem~\ref{thm:no_regret} to Appendix~\ref{appx:convex_losses}.

\begin{theorem} \label{thm:no_regret}
Let $s$ be a bounded proper scoring rule with convex exposure and forecast domain $\mathcal{D}$. For time steps $t = 1 \dots T$, an agent chooses a weight vector $\vect{w}^t \in \Delta^m$. The agent then receives a score of
\[s \parens{\sideset{}{_\vect{g}}\bigoplus_{i = 1}^m (\vect{p}_i^t, w_i^t); j^t},\]
where $\vect{p}_1^t, \dots, \vect{p}_m^t \in \mathcal{D}$ and $j^t \in [n]$ are chosen adversarially. By choosing $\vect{w}^t$ according to Algorithm~\ref{alg:ogd} (online gradient descent on the experts' weights), the agent achieves $O(\sqrt{T})$ regret in comparison with the best weight vector in hindsight. In particular, if $M$ is an upper bound\footnote{This bound $M$ exists because $s$ is bounded by assumption, and so $\vect{g}$ is also bounded (this follows from Equation~\ref{eq:s_from_g_2}).} on $\norm{\vect{g}}_2$, then for every $\vect{w}^* \in \Delta^m$ we have
\[\sum_{t = 1}^T s \parens{\sideset{}{_\vect{g}}\bigoplus_{i = 1}^m (\vect{p}_i^t, w_i^*); j^t} - s \parens{\sideset{}{_\vect{g}}\bigoplus_{i = 1}^m (\vect{p}_i^t, w_i^t); j^t} \le 3\sqrt{m}M\sqrt{T}.\]
\end{theorem}

This result is quite strong in that it does not merely achieve low regret compared to the best \emph{expert}, but in fact compared to the best possible \emph{weighted pool of experts} in hindsight. This is a substantial distinction, as it is possible for a mixture of experts to substantially outperform any individual expert.

%% file: 6_Axiomatization.tex
\section{Axiomatization of QA pooling} \label{sec:axiomatization}
In this section, we aim to show that the class of all quasi-arithmetic pooling operators is a natural one, by showing that these operators are precisely those which satisfy a natural set of axioms.

\cite{kol30} and \cite{nag30} independently considered the class of quasi-arithmetic means. Given an interval $I \subseteq \RR$ and a continuous, injective function $f: I \to \RR$, the \emph{quasi-arithmetic mean with respect to $f$}, or \emph{$f$-mean}, is the function $M_f$ that takes as input $x_1, \dots, x_m \in I$ (for any $m \ge 1$) and outputs
\[M_f(x_1, \dots, x_m) := f^{-1} \parens{\frac{f(x_1) + \dots + f(x_m)}{m}}.\]
For example, the arithmetic mean corresponds to $f(x) = x$; the quadratic to $f(x) = x^2$; the geometric to $f(x) = \log x$; and the harmonic to $f(x) = \frac{-1}{x}$.

Kolmogorov proved that the class of quasi-arithmetic means is precisely the class of functions $M: \bigcup_{m = 1}^\infty I^m \to I$ satisfying the following natural properties:\footnote{Nagumo also provided a characterization, though with slightly different properties.}

\begin{enumerate}[label=(\arabic*)]
\item $M(x_1, \dots, x_m)$ is continuous and strictly increasing in each variable.
\item $M$ is symmetric in its arguments.
\item $M(x, x, \dots, x) = x$.
\item $M(x_1, \dots, x_k, x_{k + 1}, \dots, x_m)$ = $M(y, \dots, y, x_{k + 1}, \dots, x_m)$, where $y := M(x_1, \dots, x_k)$ appears $k$ times on the right-hand side. Informally, a subset of arguments to the mean function can be replaced with their mean.
\end{enumerate}

The four properties listed above can be viewed as an \emph{axiomatization} of quasi-arithmetic means.\\

Our notion of quasi-arithmetic pooling is exactly that of a quasi-arithmetic mean, except that it is more general in two ways. First, it allows for weights to accompany the arguments to the mean. Second, we are considering quasi-arithmetic means with respect to \emph{vector-valued} functions $\vect{g}$. In the $n = 2$ outcome case, $\vect{g}$ can be considered a scalar-valued function since it is defined on a one-dimensional space (see Remark~\ref{remark:n2} for details); but in general we cannot treat $\vect{g}$ as scalar-valued.\footnote{Another difference is that in the $n = 2$ case, we are restricting $\mathcal{D}$ to be an interval from $0$ to $1$, though this is not a fundamental difference for the purposes of this section.}

Our goal is to extend the above axiomatization of quasi-arithmetic means in these two ways: first (in Section~\ref{sec:axioms_n2}) to include weights as arguments, and second (in Section~\ref{sec:axioms_general_n}) to general $n$ (while still allowing arbitrary weights).\\

\subsection{Generalizing to include weights as arguments} \label{sec:axioms_n2}
The objects that we will be studying in this section are ones of the form $(p, w)$, where $w \ge 0$ and $p \in \mathcal{D}$. In this subsection, $\mathcal{D}$ is a two-outcome forecast domain, whose elements we will identify with by the probability of the first outcome (see Remark~\ref{remark:n2}).\footnote{Proper scoring rules for two outcomes can have four possible forecast domains: $[0, 1]$, $[0, 1)$, $(0, 1]$, and $(0, 1)$.} We will fix the set $\mathcal{D}$ for the remainder of the subsection. Our results generalize to any interval of $\RR$ (as in Kolmogorov's work), but we focus on forecast domains since that is our application.

\begin{defin}
A \emph{weighted forecast} is an element of $\mathcal{D} \times \RR_{>0}$: a probability and a positive weight. Given a weighted forecast $\Pi = (p, w)$ we define $\prb(\Pi) := p$ and $\wt(\Pi) := w$.
\end{defin}

We will thinking of the output of pooling operators as weighted forecasts. This is a simple extension of our earlier definition of quasi-arithmetic pooling (Definition~\ref{def:qa_pooling}), which only output a probability.

\begin{defin}[Quasi-arithmetic pooling with arbitrary weights ($n = 2$)] \label{def:qa_ax_n2}
Given a continuous, strictly increasing function $g: \mathcal{D} \to \RR$, and weighted forecasts $\Pi_1 = (p_1, w_1), \dots, \Pi_m = (p_m, w_m)$, define the \emph{quasi-arithmetic pool} of $\Pi_1, \dots, \Pi_m$ with respect to $g$ as
\[\sideset{}{_g}\bigoplus_{i = 1}^m (p_i, w_i) := \parens{g^{-1} \parens{\frac{\sum_i w_i g(p_i)}{\sum_i w_i}}, \sum_i w_i}.\]
\end{defin}

\begin{remark}
The equivalence of this notion to our earlier one uses the fact that all (continuous) proper scoring rules for binary outcomes have convex exposure (Proposition~\ref{prop:n2_convex_exposure}), so the relationship between $g(p^*)$ and $g(p)$ may be written as in Equation~\ref{eq:g_convex_exposure}.

In the case that $\sum_i w_i = 1$, Definition~\ref{def:qa_ax_n2} reduces to Definition~\ref{def:qa_pooling}. In general, by linearly scaling the weights in Definition~\ref{def:qa_ax_n2} to add to $1$, we recover quasi-arithmetic pooling as previously defined.
\end{remark}

We find the following fact useful.

\begin{prop}
Given two continuous, strictly increasing functions $g_1$ and $g_2$, $\oplus_{g_1}$ and $\oplus_{g_2}$ are the same if and only if $g_2 = ag_1 + b$ for some $a > 0$ and $b \in \RR$.
\end{prop}

\begin{proof}
Clearly if $g_2 = ag_1 + b$ for some $a > 0$ and $b \in \RR$ then $\oplus_{g_1}$ and $\oplus_{g_2}$ are the same. For the converse, suppose that no such $a$ and $b$ exist. Let $x < y$ be such that $g_1$ and $g_2$ are not equal (even up to positive affine transformation) on $[x, y]$. Let $g_2'$ be the positive affine transformation of $g_2$ that makes it equal to $g_1$ at $x$ and $y$, and let $z \in (x, y)$ be such that $g_1(z) \neq g_2'(z)$. Let $\alpha$ be such that $g_1(z) = \alpha g_1(x) + (1 - \alpha) g_1(y)$. Then $(x, \alpha) \oplus_{g_1} (y, 1 - \alpha) = (z, 1)$, but $(x, \alpha) \oplus_{g_2} (y, 1 - \alpha) \neq (z, 1)$, so $\oplus_{g_1}$ and $\oplus_{g_2}$ are different.
\end{proof}

We now define properties (i.e. axioms) of a pooling operator $\oplus$, such that these properties are satisfied if and only if $\oplus$ is $\oplus_g$ for some $g$. Our axiomatization will look somewhat different from Kolmogorov's, in part because we choose to define $\oplus$ as a binary operator that (if it satisfies the associativity axiom) extends to the $m$-ary case. This is a simpler domain and will simplify notation. In Appendix~\ref{appx:ax} we exhibit an equivalent set of axioms that more closely resembles Kolmogorov's.

\begin{defin}[Axioms for pooling operators ($n = 2$)] \label{def:wop_n2}
For a pooling operator $\oplus$ on $\mathcal{D}$ (i.e. a binary operator on weighted forecasts), we define the following axioms.
\begin{enumerate}
\item \textbf{Weight additivity}: $\wt(\Pi_1 \oplus \Pi_2) = \wt(\Pi_1) + \wt(\Pi_2)$ for every $\Pi_1, \Pi_2$.
\item \textbf{Commutativity}: $\Pi_1 \oplus \Pi_2 = \Pi_2 \oplus \Pi_1$ for every $\Pi_1, \Pi_2$.
\item \textbf{Associativity}: $\Pi_1 \oplus (\Pi_2 \oplus \Pi_3) = (\Pi_1 \oplus \Pi_2) \oplus \Pi_3$ for every $\Pi_1, \Pi_2, \Pi_3$.
\item \textbf{Continuity}: For every $p_1, p_2$, the quantity\footnote{We allow one weight to be $0$ by defining $(p, w) \oplus (q, 0) = (q, 0) \oplus (p, w) = (p, w)$.} $\prb((p_1, w_1) \oplus (p_2, w_2))$
is a continuous function of $(w_1, w_2)$ on $\RR_{\ge 0}^2 \setminus \{(0, 0)\}$.
\item \textbf{Idempotence}: For every $\Pi_1, \Pi_2$, if $\prb(\Pi_1) = \prb(\Pi_2)$ then $\prb(\Pi_1 \oplus \Pi_2) = \prb(\Pi_1)$.
\item \textbf{Monotonicity}: Let $w > 0$ and let $p_1 > p_2 \in \mathcal{D}$. Then for $x \in (0, w)$, the quantity $\prb((p_1, x) \oplus (p_2, w - x))$ is a strictly increasing function of $x$.
\end{enumerate}
\end{defin}

The motivation for the weight additivity axiom is that the weight of a weighted forecast can be thought of as the \emph{amount of evidence} for its prediction. When pooling weighted forecasts, the weight of an individual forecast can be thought of as the strength of its vote in the aggregate.

The monotonicity axiom essentially states that if one pools two forecasts with different probabilities and a fixed total weight, then the larger the share of the weight belonging to the larger of the two probabilities, the larger the aggregate probability.\\

We now state and prove this section's main result: these axioms describe the class of QA pooling operators.

\begin{theorem} \label{thm:representation_n2}
A pooling operator is a QA pooling operator (as in Definition~\ref{def:qa_ax_n2}) with respect to some $g$ if and only if it satisfies the axioms in Definition~\ref{def:wop_n2}.\footnote{
As we mentioned, for an associative pooling operator $\oplus$, $\Pi_1 \oplus \Pi_2 \dots \oplus \Pi_m$ is a well-specified quantity, even without indicating parenthesization. This lets us use the notation $\bigoplus_{i = 1}^m \Pi_i$. This is why the statement of Theorem~\ref{thm:representation_n2} makes sense despite pooling operators not being $m$-ary by default.}
\end{theorem}

We will use $\oplus$ (without a $g$ subscript) to denote an arbitrary pooling operator that satisfies the axioms in Definition~\ref{def:wop_n2}. Before presenting the proof of Theorem~\ref{thm:representation_n2}, we will note a few important facts about weighted forecasts and pooling operators. First, we find it natural to define a notion of multiplying a weighted forecast pair by a positive constant.
	
\begin{defin}
	Given a weighted forecast $\Pi = (p, w)$ and $c > 0$, define $c\Pi := (p, cw)$.
\end{defin}
	
Note that $m\Pi = \bigoplus_{i = 1}^m \Pi$ for any positive integer $m$, by idempotence; this definition is a natural extension to all $c > 0$. We note the following (quite obvious) fact.
	
\begin{prop}
	For every weighted forecast $\Pi$ and $c_1, c_2 > 0$, we have $c_1(c_2 \Pi) = (c_1 c_2)\Pi$.
\end{prop}
	
A natural property that is not listed in Definition~\ref{def:wop_n2} is \emph{scale invariance}, i.e. that $\prb((p_1, w_1) \oplus (p_2, w_2)) = \prb((p_1, cw_1) \oplus (p_2, cw_2))$ for any positive $c$; or, equivalently, that $c(\Pi_1 \oplus \Pi_2) = c\Pi_1 \oplus c\Pi_2$. This in fact follows from the listed axioms.
	
\begin{prop}[Distributive property/scale invariance] \label{prop:dist}
	For every $\Pi_1, \Pi_2$ and any operator $\oplus$ satisfying the axioms in Definition~\ref{def:wop_n2}, we have $c(\Pi_1 \oplus \Pi_2) = c\Pi_1 \oplus c\Pi_2$.
\end{prop}

\begin{proof}
	First suppose $c$ is an integer. Then
	\[c \Pi_1 \oplus c \Pi_2 = \bigoplus_{i = 1}^c \Pi_1 \oplus \bigoplus_{i = 1}^c \Pi_2 = \bigoplus_{i = 1}^c (\Pi_1 \oplus \Pi_2) = c(\Pi_1 \oplus \Pi_2).\]
	Here, the first and last steps follow by weight additivity and idempotence. Now suppose that $c = \frac{k}{\ell}$ is a rational number. Let $\Pi_1' = \frac{1}{\ell} \Pi_1$ and $\Pi_2' = \frac{1}{\ell} \Pi_2$. We have
	\[\frac{k}{\ell}(\Pi_1 \oplus \Pi_2) = \frac{k}{\ell} (\ell \Pi_1' \oplus \ell \Pi_2') = \frac{k}{\ell} \cdot \ell(\Pi_1' \oplus \Pi_2') = k(\Pi_1' \oplus \Pi_2') = k \Pi_1' \oplus k \Pi_2' = \frac{k}{\ell} \Pi_1 \oplus \frac{k}{\ell} \Pi_2.\]
	Here, the second and second-to-last steps follow from the fact that the distributive property holds for integers.
		
	Finally, make use of the continuity axiom to extend our proof to all positive real numbers $c$. In particular, it suffices to show that $\prb(\Pi_1 \oplus \Pi_2) = \prb(c\Pi_1 \oplus c\Pi_2)$. Let $p$ be the former quantity; note that $\prb(r\Pi_1 \oplus r\Pi_2) = p$ for positive rational numbers $r$. Since the rationals are dense among the reals, it follows that for every $\epsilon > 0$, we have $\abs{\prb(c\Pi_1 \oplus c\Pi_2) - p} \le \epsilon$. Therefore, $\prb(c\Pi_1 \oplus c\Pi_2) = p$. This completes the proof.
\end{proof}

Armed with these facts, we present a proof of Theorem~\ref{thm:representation_n2}.

\begin{proof}[Proof of Theorem~\ref{thm:representation_n2}]
	We first prove that any QA pooling operator $\oplus_g$ satisfies the axioms in Definition~\ref{def:wop_n2}. Weight additivity, commutativity, and idempotence are trivial. Associativity is also clear: given $\Pi_1 = (p_1, w_1)$ and likewise $\Pi_2, \Pi_3$, we have
	\begin{align*}
		g(\prb((\Pi_1 \oplus_g \Pi_2) \oplus_g \Pi_3)) &= \frac{(w_1 + w_2) g(\prb(\Pi_1 \oplus_g \Pi_2)) + w_3 g(p_3)}{(w_1 + w_2) + w_3}\\
		&= \frac{(w_1 + w_2) \frac{w_1 g(p_1) + w_2 (p_2)}{w_1 + w_2} + w_3 g(p_3)}{w_1 + w_2 + w_3} = \frac{w_1 g(p_1) + w_2 g(p_2) + w_3 g(p_3)}{w_1 + w_2 + w_3}
	\end{align*}
	and likewise for $g(\prb(\Pi_1 \oplus_g (\Pi_2 \oplus_g \Pi_3)))$, so $\prb((\Pi_1 \oplus_g \Pi_2) \oplus_g \Pi_3) = \prb(\Pi_1 \oplus_g (\Pi_2 \oplus_g \Pi_3))$ (since $g$ is strictly increasing and therefore injective). The fact that the weights are also the same is trivial. Continuity follows from the fact that
	\[\prb(\Pi_1 \oplus_g \Pi_2) = g^{-1} \parens{\frac{w_1 g(p_1) + w_2 g(p_2)}{w_1 + w_2}}\]
	is continuous in $(w_1, w_2)$ (when $w_1, w_2$ are not both zero). Here we are using the fact that $g$ is \emph{strictly} increasing, which means that $g^{-1}$ is continuous.
	
	Finally, regarding the monotonicity axiom, for any fixed $w$ and $p_1 > p_2$ (as in the axiom statement), we have
	\[g(\prb((p_1, x) \oplus_g (p_2, w - x))) = \frac{x g(p_1) + (w - x) g(p_2)}{x + w - x} = \frac{x g(p_1) + (w - x) g(p_2)}{w}.\]
	Since $p_1 > p_2$, we have $g(p_1) > g(p_2)$, so the right-hand side strictly increases with $x$. Since $g^{-1}$ is also strictly increasing, it follows that $\prb((p_1, x) \oplus_g (p_2, w - x))$ strictly increases with $x$.\\
	
	The converse --- that every pooling operator satisfying the axioms in Definition~\ref{def:wop_n2} is $\oplus_g$ for some $g$ --- works by constructing $g$ by fixing it at two points and constructing $g$ at all other points. Right now we show how to do this when the forecast domain is $[0, 1]$; see the proof of Theorem~\ref{thm:representation} for the argument in full generality.
	
	Let $\oplus$ be a pooling operator that satisfies our axioms. Define $g$ as follows: let $g(0) = 0$ and $g(1) = 1$. For $0 < p < 1$, define $g(p) = w$ where $(1, w) \oplus (0, 1 - w) = (p, 1)$. (This $w$ exists by continuity and the intermediate value theorem; it is unique by the ``strictly" increasing stipulation of monotonicity.) Note that $g$ is continuous and increasing by monotonicity.\footnote{As a matter of fact, $g$ is strictly increasing because it is impossible for $g(p_1)$ to equal $g(p_2)$ for $p_1 \neq p_2$, as that would mean that $(1, g(p_1)) \oplus (0, 1 - g(p_1)) = (p_1, 1) = (p_2, 1)$. Another way to look at this is that it comes from the fact that $(1, w) \oplus (0, 1 - w)$ is continuous in $w$ by the continuity axiom. In a sense, the \emph{continuity} of $g$ corresponds to the \emph{strictness} of increase in the monotonicity axiom and the \emph{strictness} of increase of $g$ corresponds to the \emph{continuity} axiom.}
	
	We wish to show that for any $\Pi_1 = (p_1, w_1)$ and $\Pi_2 = (p_2, w_2)$, we have that $\Pi_1 \oplus \Pi_2 = \Pi_1 \oplus_g \Pi_2$. Clearly the weight of both sides is $w_1 + w_2$, so we wish to show that the probabilities on each side are the same. We have\footnote{Steps 3 and 7 uses the distributive property (Proposition~\ref{prop:dist}).}
	\begin{align*}
		\prb(\Pi_1 \oplus \Pi_2) &= \prb(w_1 (p_1, 1) \oplus w_2 (p_2, 1))\\
		&= \prb(w_1 ((1, g(p_1)) \oplus (0, 1 - g(p_1))) \oplus w_2 ((1, g(p_2)) \oplus (0, 1 - g(p_2))))\\
		&= \prb(w_1 (1, g(p_1)) \oplus w_1 (0, 1 - g(p_1)) \oplus w_2 (1, g(p_2)) \oplus w_2 (0, 1 - g(p_2)))\\
		&= \prb((1, w_1 g(p_1)) \oplus (0, w_1(1 - g(p_1))) \oplus (1, w_2 g(p_2)) \oplus (0, w_2(1 - g(p_2))))\\
		&= \prb((1, w_1 g(p_1) + w_2 g(p_2)) \oplus (0, w_1(1 - g(p_1)) + w_2(1 - g(p_2))))\\
		&= \prb \parens{\frac{1}{w_1 + w_2}((1, w_1 g(p_1) + w_2 g(p_2)) \oplus (0, w_1(1 - g(p_1)) + w_2(1 - g(p_2))))}\\
		&= \prb \parens{\parens{1, \frac{w_1 g(p_1) + w_2 g(p_2)}{w_1 + w_2}} \oplus \parens{0, \frac{w_1(1 - g(p_1)) + w_2(1 - g(p_2))}{w_1 + w_2}}},
	\end{align*}
	which by definition of $g$ is equal to the probability $p$ such that $g(p) = \frac{g(p_1) w_1 + g(p_2) w_2}{w_1 + w_2}$. That is, $\prb(\Pi_1 \oplus \Pi_2) = \prb(\Pi_1 \oplus_g \Pi_2)$.
	
	Showing that $\oplus$ and $\oplus_g$ are equivalent for more than two arguments is now trivial:
	\[\sideset{}{_\vect{g}}\bigoplus_{i = 1}^m \Pi_i = \Pi_1 \oplus_g \Pi_2 \oplus_g \Pi_3 \dots \oplus_g \Pi_m = \Pi_1 \oplus \Pi_2 \oplus_g \Pi_3 \dots \oplus_g \Pi_m = \dots = \bigoplus_{i = 1}^m \Pi_i.\]
	(Here we are implicitly using the fact that $\oplus_g$ is associative, as we proved earlier.) This completes the proof.
\end{proof}

\subsection{Generalizing to higher dimensions} \label{sec:axioms_general_n}
In Appendix~\ref{appx:extending}, we generalize our axioms from $n = 2$ outcomes to arbitrary values of $n$. An important challenge is extending the monotonicity axiom: in higher dimensions, what is the appropriate generalization of an increasing function? We show that the correct notion is \emph{cyclical monotonicity}, which we define and discuss. We then present our axiomatization (Definition~\ref{def:qa_ax}) and prove that the axioms constitute a characterization of the class of QA pooling operators (Theorem~\ref{thm:representation}). On a high level, the proof is not dissimilar to that of Theorem~\ref{thm:representation_n2}, though the details are fairly different and more technical.\\

In conclusion, in Definition~\ref{def:wop_n2} we made a list of natural properties that a pooling operator may satisfy. Theorem~\ref{thm:representation_n2} shows that the pooling operators satisfying these properties are exactly the QA pooling operators. In Appendix~\ref{appx:extending}, we generalize this theorem to higher dimensions, thus fully axiomatizing QA pooling. This result gives us an additional important reason to believe that QA pooling with respect to a proper scoring rule is a fundamental notion.

%% file: 8_Conclusion.tex
\section{Conclusions and future work} \label{sec:conclusion}

While in this work we have focused on scoring rules for probability distributions over outcomes, scoring rules can be used to elicit various other properties, such as the expectation or the median of a random variable \cite{sav71} \cite{gr07} \cite{lps08}. The topic of \emph{property elicitation} studies such scoring rules. The representation of a proper scoring rule in terms of its expected reward function (\cite[Theorem 2]{gr07}, our Proposition~\ref{prop:G_convex}) generalizes to arbitrary linear properties \cite[Theorem 11]{fk15}. That is, consider a random variable (or $n$-tuple of random variables) $X$ and a convex function $G$ whose domain is (a superset of) the possible values of $X$. If $\vect{p}$ is an expert's forecast for $\EE{X}$ and $\vect{x}$ is the realized outcome, then the scoring rule $s(\vect{p}; \vect{x}) := G(\vect{p}) + \angles{\vect{g}(\vect{p}), \vect{x} - \vect{p}}$ is proper.\footnote{In the setting of this paper, $X$ is the vector of random variables, where the $j$-th variable is $1$ if outcome $j$ happens and $0$ otherwise; and $\vect{x} = \delta_j$ where $j$ is the outcome that happens.} Modulo a fairly straightforward generalization, all proper scoring rules for $\EE{X}$ take this form.

Our definition of QA pooling extends verbatim to the setting of eliciting linear properties. In this more general setting, for any proper scoring rule that has convex exposure (see Definition~\ref{def:convex_exposure}), it remains the case $\vect{g}(\vect{p}^*) = \sum_i w_i \vect{g}(\vect{p}_i)$; as a consequence, Theorems~\ref{thm:max_min}, \ref{thm:ws_concave}, and \ref{thm:no_regret} generalize.\\

There are several other interesting directions for generalization and future work:

\begin{itemize}
\item As mentioned in Section~\ref{sec:related_work}, QA pooling can be interpreted in terms of cost function markets: if $\vect{q}_i$ denotes the quantity of shares that each expert would buy in order to bring a market into line with the experts' probabilities, then the QA pool is the probability corresponding to the quantity vector $\vect{q}^* := \sum_i w_i \vect{q}_i$. This follows from \cite[Eq. 22]{acv13}; the underlying reason is the convex duality between an expected reward function $G$ and the corresponding cost function $C$. Exploring this connection further may yield insights into forecast aggregation through market mechanisms.
\item While idempotence is a fairly natural axiom, Bayesian approaches to aggregation often violate this axiom: for example, an aggregator who receives two independent estimates with the same probability should act differently than an aggregator who only received one of the estimates. It would be interesting to explore whether our results and techniques are applicable to notions of pooling that do not satisfy idempotence, or to other (perhaps Bayesian) settings.
\item As we discussed in Section~\ref{sec:related_work}, there is a fair amount of work on aggregating forecasts with prediction markets --- often ones that are based on proper scoring rules. Is there a natural trading-based interpretation of QA pooling?
%\item Definition~\ref{def:generalized_qa_pool} gave a natural generalization of QA pooling to proper scoring rules that do not have convex exposure. There is another potentially natural generalization, which is to define the QA pool as the forecast $\vect{p}^*$ maximizing $\min_j u(\vect{p}^*; j)$ (as defined in Theorem~\ref{thm:max_min}). Is this generalization equivalent to Definition~\ref{def:generalized_qa_pool}? If not, how does it behave?
\item Although our proof of Theorem~\ref{thm:ws_concave} (that the QA pool of forecasts is concave in the experts' weights) relies on the convex exposure property, we have not ruled out the possibility that the result holds even without this assumption. Is this the case? Even if not, are no-regret algorithms for learning weights still possible?
\item Our no-regret algorithm for learning weights relies on $s$ being bounded, because this allows us to place a concrete upper bound on $\norm{\nabla L^t(\cdot)}_2$. If $s$ is unbounded, in a fully adversarial setting it is impossible to achieve a no-regret guarantee. On the other hand, \cite{nr22} show that if experts are assumed to be calibrated, then learning the optimal weights for logarithmic pooling is possible when using the log scoring rule. Does this result generalize to learning optimal weights for QA pooling with respect to other unbounded proper scoring rules?
\item We have presented a list of axioms characterizing QA pooling operators. Is there an alternative axiomatization that uses equally natural but fewer axioms?
\item Do our results extend to probability distributions over infinitely many outcomes (e.g.\ for experts who wish to forecast a real number)?
\end{itemize}

%% file: Appendix_Intro.tex
\section{Details omitted from Section~\ref{sec:intro}} \label{appx:intro}
\paragraph{Logarithmic pooling as averaging Bayesian evidence} We discuss for simplicity the binary outcome case, though this discussion holds in general. Suppose that an expert assigns a probability to an event $X$ occurring by updating on some prior (50\%, say, though this does not matter). Suppose that the expert receives evidence $E$. Bayesian updating works as such:
\[\frac{\pr{X \mid E}}{\pr{\neg X \mid E}} = \frac{\pr{X}}{\pr{\neg X}} \cdot \frac{\pr{E \mid X}}{\pr{E \mid \neg X}}.\]
That is, the expert \emph{multiplies their odds of $X$} (i.e. the probability of $X$ divided by the probability of $\neg X$) by the relative likelihood that $E$ would be the case conditioned on $X$ versus $\neg X$. Equivalently, we can take the log of both sides; this tells us that the posterior log odds of $X$ is equal to the prior log odds plus $\log \frac{\pr{E \mid X}}{\pr{E \mid \neg X}}$. For every piece of evidence $E_k$ that the expert receives, they make this update (assuming that the $E_k$'s are mutually independent conditioned on $X$).

This means that for any $E_k$, we can view the quantity $\log \frac{\pr{E_k \mid X}}{\pr{E_k \mid \neg X}}$ as the strength of evidence that $E_k$ gives in favor of $X$. We call this the ``Bayesian evidence" in favor of $X$ given by $E_k$. The \emph{total Bayesian evidence} that the expert has in favor of $X$ (i.e. the sum of these values over all $E_k$) is the expert's log odds of $X$ (i.e. $\log \frac{\pr{X}}{\pr{\neg X}}$).

Logarithmic pooling takes the average of experts' log odds of $X$. As we have shown, if the experts are Bayesian, this amounts to taking the average of all experts' amounts of Bayesian evidence in favor of $X$.

%% file: Appendix_Convex_Losses.tex
\section{Details omitted from Section~\ref{sec:convex_losses}} \label{appx:convex_losses}
\begin{proof}[Proof of Theorem~\ref{thm:ws_concave}]
	Let $\vect{v}$ and $\vect{w}$ be two weight vectors. We wish to show that for any $c \in [0, 1]$, we have
	\[\ws_j(c\vect{v} + (1 - c)\vect{w}) - c \ws_j(\vect{v}) - (1 - c) \ws_j(\vect{w}) \ge 0.\]
	Recall the notation $\vect{p}^*_{\vect{w}}$ from Definition~\ref{def:qa_pooling}. Note that
	\begin{equation} \label{eq:g_linear}
		\vect{g}(\vect{p}^*_{c\vect{v} + (1 - c)\vect{w}}) = \sum_{i = 1}^m (cv_i + (1 - c)w_i) \vect{g}(\vect{p}_i) = c \vect{g}(\vect{p}^*_{\vect{v}}) + (1 - c) \vect{g}(\vect{p}^*_{\vect{w}}).
	\end{equation}
	We have
	\begin{align*}
		&\quad \ws_j(c\vect{v} + (1 - c)\vect{w}) - c \ws_j(\vect{v}) - (1 - c) \ws_j(\vect{w})\\
		&= s(\vect{p}^*_{c\vect{v} + (1 - c)\vect{w}}; j) - c s(\vect{p}^*_{\vect{v}}; j) - (1 - c) s(\vect{p}^*_{\vect{w}}; j)\\
		&= G(\vect{p}^*_{c\vect{v} + (1 - c)\vect{w}}) + \angles{\vect{g}(\vect{p}^*_{c\vect{v} + (1 - c)\vect{w}}), \delta_j - \vect{p}^*_{c\vect{v} + (1 - c)\vect{w}}}\\
		&\qquad - c(G(\vect{p}^*_{\vect{v}}) + \angles{\vect{g}(\vect{p}^*_{\vect{v}}), \delta_j - \vect{p}^*_{\vect{v}}}) - (1 - c)(G(\vect{p}^*_{\vect{w}}) + \angles{\vect{g}(\vect{p}^*_{\vect{w}}), \delta_j - \vect{p}^*_{\vect{w}}})\\
		&= G(\vect{p}^*_{c\vect{v} + (1 - c)\vect{w}}) - \angles{\vect{g}(\vect{p}^*_{c\vect{v} + (1 - c)\vect{w}}), \vect{p}^*_{c\vect{v} + (1 - c)\vect{w}}}\\
		&\qquad - c(G(\vect{p}^*_{\vect{v}}) - \angles{\vect{g}(\vect{p}^*_{\vect{v}}), \vect{p}^*_{\vect{v}}}) - (1 - c)(G(\vect{p}^*_{\vect{w}}) - \angles{\vect{g}(\vect{p}^*_{\vect{w}}), \vect{p}^*_{\vect{w}}})
	\end{align*}
	Step 1 follows from the definition of $\ws$. Step 2 follows from Equation~\ref{eq:s_from_g_2}. Step 3 follows from Equation~\ref{eq:g_linear}, and specifically that the inner product of each side with $\delta_j$ is the same (so the $\delta_j$ terms cancel out, leaving a quantity that does not depend on $j$). Continuing where we left off:
	
	\begin{align*}
		&\quad \ws_j(c\vect{v} + (1 - c)\vect{w}) - c \ws_j(\vect{v}) - (1 - c) \ws_j(\vect{w})\\
		&= G(\vect{p}^*_{c\vect{v} + (1 - c)\vect{w}}) - c \angles{\vect{g}(\vect{p}^*_{\vect{v}}), \vect{p}^*_{c\vect{v} + (1 - c)\vect{w}}} - (1 - c) \angles{\vect{g}(\vect{p}^*_{\vect{w}}), \vect{p}^*_{c\vect{v} + (1 - c)\vect{w}}}\\
		&\qquad - c(G(\vect{p}^*_{\vect{v}}) - \angles{\vect{g}(\vect{p}^*_{\vect{v}}), \vect{p}^*_{\vect{v}}}) - (1 - c)(G(\vect{p}^*_{\vect{w}}) - \angles{\vect{g}(\vect{p}^*_{\vect{w}}), \vect{p}^*_{\vect{w}}})\\
		&= c(G(\vect{p}^*_{c\vect{v} + (1 - c)\vect{w}}) - G(\vect{p}^*_{\vect{v}}) - \angles{\vect{g}(\vect{p}^*_{\vect{v}}), \vect{p}^*_{c\vect{v} + (1 - c)\vect{w}} - \vect{p}^*_{\vect{v}}})\\
		&\qquad + (1 - c)(G(\vect{p}^*_{c\vect{v} + (1 - c)\vect{w}}) - G(\vect{p}^*_{\vect{w}}) - \angles{\vect{g}(\vect{p}^*_{\vect{w}}), \vect{p}^*_{c\vect{v} + (1 - c)\vect{w}} - \vect{p}^*_{\vect{w}}})\\
		&= c D_G(\vect{p}^*_{c\vect{v} + (1 - c)\vect{w}} \parallel \vect{p}^*_{\vect{v}}) + (1 - c) D_G(\vect{p}^*_{c\vect{v} + (1 - c)\vect{w}} \parallel \vect{p}^*_{\vect{w}}) \ge 0.
	\end{align*}
	Step 4 again follows from Equation~\ref{eq:g_linear}. Step 5 is a rearrangement of terms. Finally, step 6 follows from the definition of Bregman divergence, and step 7 follows from the fact that Bregman divergence is always non-negative. This completes the proof.
\end{proof}

\begin{remark}
Theorem~\ref{thm:ws_concave} can be stated in more generality: $s$ need not have convex exposure; it suffices to have that for the particular $\vect{p}_1, \dots, \vect{p}_m$, the QA pool of these forecasts exists for every weight vector.
\end{remark}

Beyond Theorem~\ref{thm:ws_concave}'s instrumental use for no-regret online learning of expert weights (Theorem~\ref{thm:no_regret} below), the result is interesting in its own right. For example, the following fact --- loosely speaking, that QA pooling cannot benefit from weight randomization --- follows as a corollary. (Recall the definition of $\vect{p}_{\vect{w}}^*$ from Definition~\ref{def:qa_pooling}.)

\begin{corollary}
Consider a randomized algorithm $A$ with the following specifications:
\begin{itemize}
	\item Input: a proper scoring rule $s$ with convex exposure, expert forecasts $\vect{p}_1, \dots, \vect{p}_m$.
	\item Output: a weight vector $\vect{w} \in \Delta^m$.
\end{itemize}
For any input $s, \vect{p}_1, \dots, \vect{p}_m$ and for every $j$, we have
\[s(\vect{p}^*_{\EE[A]{\vect{w}}}; j) \ge \EE[A]{s(\vect{p}^*_{\vect{w}}; j)},\]
where $\vect{p}^*_{\vect{x}}$ denotes the QA pool of $\vect{p}_1, \dots, \vect{p}_m$ with weight vector $\vect{x}$.
\end{corollary}

\begin{myalgorithm}[Online gradient descent algorithm for Theorem~\ref{thm:no_regret}] \label{alg:ogd}
	We proceed as follows:
	\begin{itemize}
		\item For $t \ge 1$, define $\eta_t := \frac{1}{M\sqrt{mt}}$.
		\item Start with an arbitrary guess $\vect{w}^1 \in \Delta^m$.
		\item At each time step $t$ from $1$ to $T$:
		\begin{itemize}
			\item Play $\vect{w}^t$ and observe loss $L^t(\vect{w}^t)$.
			\item Let $\tilde{\vect{w}}^{t + 1} = \vect{w}^t - \eta_t \nabla L^t(\vect{w}^t)$. If $\tilde{\vect{w}}^{t + 1} \in \Delta^m$, let $\vect{w}^{t + 1} = \tilde{\vect{w}}^{t + 1}$. Otherwise, let $\vect{w}^{t + 1}$ be the orthogonal projection of $\tilde{\vect{w}}^{t + 1}$ onto $\Delta^m$.
		\end{itemize}
	\end{itemize}
\end{myalgorithm}

The above algorithm is an adaptation of the online gradient descent algorithm (as presented in \cite[\S3.1]{hazan19}) for the setting of Theorem~\ref{thm:no_regret}. In our context, the loss function is $L^t(\vect{w}) = -\ws_{j^t}(\vect{w})$, where $\ws$ is as in Theorem~\ref{thm:ws_concave}, relative to forecasts $\vect{p}_1^t, \dots, \vect{p}_m^t$.

Below, we prove that Theorem~\ref{thm:no_regret}, i.e.\ that Algorithm~\ref{alg:ogd} satisfies the details of Theorem~\ref{thm:no_regret}. The proof amounts to applying the standard bounds for online gradient descent, though with an extra step: we use the bound $M$ on $\norm{\vect{g}}$ to bound the gradient of the loss as a function of expert weights.

\begin{proof}[Proof of Theorem~\ref{thm:no_regret}]
	We apply \cite[Theorem 3.1]{hazan19}; this theorem tells us that in order to prove the stated bound, it suffices to show that for all $t$ and $\vect{w}$, $\norm{\nabla L^t(\vect{w})}_2 \le \sqrt{2m}M$.
	
	Let $L$ be an arbitrary loss function, i.e. $L(\vect{w}) = -\ws_j(\vect{w})$ for some $j, \vect{p}_1, \dots, \vect{p}_m$. Let $\vect{p}^*(\vect{w}) = \sideset{}{_\vect{g}}\bigoplus\limits_{i = 1}^m (\vect{p}_i, w_i)$. We claim that
	\begin{equation} \label{eq:delLw}
		\nabla L(\vect{w}) = \begin{pmatrix} \vect{g}(\vect{p}_1) \\ \vdots \\ \vect{g}(\vect{p}_m) \end{pmatrix} (\vect{p}^*(\vect{w}) - \delta_j),
	\end{equation}
	where this $m$-dimensional vector should be interpreted modulo translation by $\vect{1}_m$ (see Remark~\ref{remark:mod_T1n}). To see this, observe that
	\[\nabla L(\vect{w}) = - \nabla WS_j(\vect{w}) = - \nabla_{\vect{w}} s(\vect{p}^*(\vect{w}); j) = - \nabla_{\vect{w}}(G(\vect{p}^*(\vect{w})) + \angles{\vect{g}(\vect{p}^*(\vect{w})), \delta_j - \vect{p}^*(\vect{w})}),\]
	where $\nabla_{\vect{w}}$ denotes the gradient with respect to change in the weight vector $\vect{w}$ (as opposed to change in the probability vector). Now, by the chain rule for gradients, we have
	\[\nabla_{\vect{w}} G(\vect{p}^*(\vect{w})) = (J_{\vect{p}^*}(\vect{w}))^{\top} \vect{g}(\vect{p}^*(\vect{w})),\]
	where $J_{\vect{p}^*}$ denotes the Jacobian matrix of the function $\vect{p}^*(\vect{w})$. Also, we have
	\[\vect{g}(\vect{p}^*(\vect{w})) = \sum_{i = 1}^m w_i \vect{g}(\vect{p}_i),\]
	so (again by the chain rule) we have
	\[\nabla_{\vect{w}}(\angles{\vect{g}(\vect{p}^*(\vect{w})), \delta_j - \vect{p}^*(\vect{w})}) = \begin{pmatrix} \vect{g}(\vect{p}_1) \\ \vdots \\ \vect{g}(\vect{p}_m) \end{pmatrix} (\delta_j - \vect{p}^*(\vect{w})) - (J_{\vect{p}^*}(\vect{w}))^{\top} \vect{g}(\vect{p}^*(\vect{w})).\]
	This gives us Equation~\ref{eq:delLw}.\footnote{Note that the cancellation of the Jacobian terms stems not from the specific relationship between $\vect{p}^*$ and $\vect{w}$ but from the nature of proper scoring rules. We obtain the same cancellation if we consider $\nabla s(\vect{p}; j)$, where after differentiating $G(\vect{p}) + \angles{\vect{g}(\vect{p}), \delta_j - \vect{p}}$ we find that the $\vect{g}(\vect{p})$ terms cancel.}\\
	
	Now, for any $i$, we have
	\[\abs{\angles{\vect{g}(\vect{p}_i), \vect{p}^*(\vect{w}) - \delta_j}} \le \norm{\vect{g}(\vect{p}_i)}_2 \norm{\vect{p}^*(\vect{w}) - \delta_j}_2 \le \sqrt{2}M.\]
	Therefore,
	\[\norm{\nabla L(\vect{w})}_2 \le \sqrt{m \cdot (\sqrt{2}M)^2} = \sqrt{2m}M,\]
	completing the proof.
\end{proof}

%% file: Appendix_Axioms.tex
\section{Details omitted from Section~\ref{sec:axiomatization}} \label{appx:ax}
	We claim that our axioms in Definition~\ref{def:qa_ax_n2} can be restated equivalently in a form similar to that of Kolmogorov introduced at the top of Section~\ref{sec:axiomatization} (though with weights.)
	\begin{claim}
		Given a pooling operator $\oplus$ on $\mathcal{D}$ satisfying Definition~\ref{def:qa_ax_n2}, the function $M$ defined on arbitrary tuples of weighted forecasts defined by $M(\Pi_1, \dots, \Pi_m) := \bigoplus_{i = 1}^m \Pi_i$ satisfies the following axioms:
		\begin{enumerate}[label=(\arabic*)]
			\item $M(\Pi_1, \dots, \Pi_m)$ is strictly increasing in each $\prb(\Pi_i)$ and continuous in its inputs.\footnote{That is, it is a continuous function of its input in $\mathcal{D}^m \times (\RR_{\ge 0}^m \setminus \vect{0})$, where weighted forecasts with weight $0$ are ignored when computing $M$.}
			\item $M$ is symmetric in its arguments.
			\item $M((p, w_1), \dots, (p, w_m)) = (p, \sum_i w_i)$.
			\item $M(\Pi_1, \dots, \Pi_k, \Pi_{k + 1}, \dots, \Pi_m) = M(\Pi', \Pi_{k + 1}, \dots, \Pi_m)$, where $y := M(\Pi_1, \dots, \Pi_k)$.
			\item $M((p_1, w_1), \dots, (p_m, w_m))$ has weight $w_1 + \dots + w_m$.
		\end{enumerate}
		Additionally, given any $M$ defined on arbitrary tuples of weighted forecasts, the operator $\oplus$ defined by $\Pi_1 \oplus \Pi_2 := M(\Pi_1, \Pi_2)$ satisfies Definition~\ref{def:qa_ax_n2}.
	\end{claim}

	\begin{proof}
		We first prove that given $\oplus$ satisfying Definition~\ref{def:qa_ax_n2}, $M$ satisfies the stated axioms. The last four axioms are clear, so we prove the first one. The fact that $M$ is strictly increasing in each probability follows immediately by considering the continuous, strictly increasing function $g$ such that $\oplus = \oplus_g$, which exists by Theorem~\ref{thm:representation_n2}. Continuity likewise follows, since the quantity in Definition~\ref{def:qa_ax_n2} is continuous.
		
		We now prove that given $M$ satisfying the stated axioms, $\oplus$ satisfies Definition~\ref{def:qa_ax_n2}. Weight additivity, commutativity, continuity, and idempotence are clear. To prove associativity, note that
		\[\Pi_1 \oplus (\Pi_2 \oplus \Pi_3) = M(\Pi_1, M(\Pi_2, \Pi_3)) = M(\Pi_1, \Pi_2, \Pi_3) = M(M(\Pi_1, \Pi_2), \Pi_3) = (\Pi_1 \oplus \Pi_2) \oplus \Pi_3.\]
		To prove monotonicity, let $p_1 > p_2$ and $w > x > y$. We wish to prove that $\prb((p_1, x) \oplus (p_2, w - x)) > \prb((p_1, y) \oplus (p_2, w - y))$. We have
		\begin{align*}
			\prb((p_1, x) \oplus (p_2, w - x)) &= \prb(M((p_1, x), (p_2, w - x))) = \prb(M((p_1, y), (p_1, x - y), (p_2, w - x)))\\
			&> \prb(M((p_1, y), (p_2, x - y), (p_2, w - x))) = \prb(M((p_1, y), (p_2, w - y)))\\
			&= \prb((p_1, y) \oplus (p_2, w - y)).
		\end{align*}
	\end{proof}

\subsection{Extending the results of Section~\ref{sec:axiomatization} to $n > 2$ outcomes} \label{appx:extending}
We now discuss extending our axiomatization to arbitrary values of $n$ in a way that, again, describes the class of QA pooling operators. Just as we fixed a two-outcome forecast domain $\mathcal{D}$ in Section~\ref{sec:axiomatization}, we now fix an $n$-outcome forecast domain $\mathcal{D}$ for any $n \ge 2$. Our definition of weighted forecasts remains the same (except that now $\prb(\Pi)$ is a vector). Our definition of quasi-arithmetic pooling, however, needs to change to make $\vect{g}$ vector-valued. This raises the question: what is the analogue of ``increasing" for vector-valued functions? It turns out that the relevant notion for us is \emph{cyclical monotonicity}, introduced by Rockafellar \cite{roc70_paper} (see also \cite[\S27]{roc70_book}). We will define this notion shortly, but first we give the definition of quasi-arithmetic pooling with arbitrary weights (analogous to Definition~\ref{def:qa_ax_n2}) for this setting. Throughout this section, we will use the notation $H_n(c) := \{\vect{x} \in \RR^n: \sum_i x_i = c\}$. Recall from Remark~\ref{remark:mod_T1n} that the range of the gradient of a function defined on $\mathcal{D}$ is a subset of $H_n(0)$.

\begin{defin}[Quasi-arithmetic pooling with arbitrary weights] \label{def:qa_ax}
	Given a continuous, strictly cyclically monotone vector-valued function $\vect{g}: \mathcal{D} \to H_n(0)$ whose range is a convex set, and weighted forecasts $\Pi_1 = (\vect{p}_1, w_1), \dots, \Pi_m = (\vect{p}_m, w_m)$, define the \emph{quasi-arithmetic pool} of $\Pi_1, \dots, \Pi_m$ with respect to $\vect{g}$ as
	\[\sideset{}{_{\vect{g}}}\bigoplus_{i = 1}^m (\vect{p}_i, w_i) := \parens{\vect{g}^{-1} \parens{\frac{\sum_i w_i \vect{g}(\vect{p}_i)}{\sum_i w_i}}, \sum_i w_i}.\]
\end{defin}

Note that QA pooling as defined in Definition~\ref{def:qa_pooling} can be written in the form of Definition~\ref{def:qa_ax} if and only if the scoring rule has convex exposure; if it does not, then for some choices of parameters, $\sum_i w_i \vect{g}(\vect{p}_i)$ will be equal to a subgradient -- but not the gradient -- of $G$ at some point.

\begin{defin}[Cyclical monotonicity] \label{def:cyc_mon}
	A function $\vect{g}: U \subseteq \RR^n \to \RR^n$ is \emph{cyclically monotone} if for every list of points $\vect{x}_0, \vect{x}_1, \dots, \vect{x}_{k - 1}, \vect{x}_k = \vect{x}_0 \in U$, we have
	\[\sum_{i = 1}^k \angles{\vect{g}(\vect{x}_i), \vect{x}_i - \vect{x}_{i - 1}} \ge 0.\]
	We also say that $\vect{g}$ is \emph{strictly cyclically monotone} if the inequality is strict except when $\vect{x}_0 = \dots = \vect{x}_{k - 1}$.
\end{defin}

To gain an intuition for this notion, consider the case of $k = 2$; then this condition says that $\angles{\vect{g}(\vect{x}_1) - \vect{g}(\vect{x}_0), \vect{x}_1 - \vect{x}_0} \ge 0$. In other words, the change in $\vect{g}$ from $\vect{x}_0$ to $\vect{x}_1$ is in the same general direction as the direction from $\vect{x}_0$ to $\vect{x}_1$. This property is called \emph{2-cycle} (or \emph{weak}) \emph{monotonicity}.

Cyclical monotonicity is a stronger notion, which may be familiar to the reader for its applications in mechanism design and revealed preference theory, see e.g. \cite{ls07}, \cite{abhm10}, \cite{fk14}, \cite[\S2.1]{vohra07}. In such settings, it is usually the case that two-cycle and cyclical monotonicity are equivalent. Indeed, Saks and Yu showed that these conditions are equivalent in settings where the set of outcomes (i.e. the range of $\vect{g}$) is finite \cite{sy05}. However, cyclical monotonicity is substantially stronger than two-cycle monotonicity when the range of $\vect{g}$ is infinitely large, as in our setting. In fact, the difference between these two conditions is that a two-cycle monotone function is cyclically monotone if and only if it is also \emph{vortex-free} \cite[Theorem 3.9]{ak14}. \emph{Vortex-freeness} means that the path integral of $\vect{g}$ along any triangle vanishes. See \cite{ak14} for a deteailed comparison of these two notions.\\

The immediately relevant fact for us is that cyclically monotone functions are gradients of convex functions (and vice versa). Speaking more precisely:

\begin{theorem} \label{thm:cyc_mon_convex}
	A vector-valued function $\vect{g}$ is continuous and strictly cyclically monotone if and only if it is the gradient of a differentiable, strictly convex function $G$.
\end{theorem}

\begin{proof}[Proof of Theorem~\ref{thm:cyc_mon_convex}]
	Per a theorem of Rockafellar (\cite{roc70_paper}, see also Theorem 24.8 in \cite{roc70_book}), a function $\vect{g}$ is cyclically monotone if and only if it is a subgradient of a convex function $G$. The proof of this fact shows just as easily that a function is strictly cyclically monotone if and only if it is a subgradient of a strictly convex function.
	
	Consider a differentiable, strictly convex function $G$. Its gradient is continuous (see \cite[Theorem 25.5]{roc70_book}). Conversely, consider a continuous, strictly cyclically monotone vector-valued function $\vect{g}$. As we just discussed, it is a subgradient of some strictly convex function $G$. A convex function with a continuous subgradient is differentiable \cite[Proposition 17.41]{bc11}.
\end{proof}

This means that the conditions on $\vect{g}$ in Definition~\ref{def:qa_ax} are precisely those necessary to let $\vect{g}$ be any function that it could be in our original definition of quasi-arithmetic pooling (Definition~\ref{def:qa_pooling}). Our new definition is thus equivalent to the old one (after normalizing weights to add to $1$).\\

We now discuss our axioms for pooling operators that will again capture the class of QA pooling operators. We will keep the weight additivity, commutativity, associativity, and idempotence verbatim from our discussion of the $n = 2$ case. We will slightly strengthen the continuity argument (see below).

We will also add a new axiom, \emph{subtraction}, which states that if $\Pi_1 \oplus \Pi_2 = \Pi_1 \oplus \Pi_3$ then $\Pi_2 = \Pi_3$. Subtraction in the $n = 2$ case follows from monotonicity; in this case, however, we the subtraction axiom will help us state the monotonicity axiom. In particular, it allows us to make the following definition, which essentially extends the notion of pooling to allow for negative weights.

\begin{defin} \label{def:p}
	Let $\oplus$ be a pooling operator satisfying weight additivity, commutativity, associativity, and subtraction. Fix $\vect{p}_1, \dots, \vect{p}_k \in \mathcal{D}$. Define a function $\vect{p}: \Delta^k \to \mathcal{D}$ (with $\vect{p}_1, \dots, \vect{p}_k$ serving as implicit arguments) defined by
	\[\vect{p}(w_1, \dots, w_k) = \prb \parens{\bigoplus_{i = 1}^k (\vect{p}_i, w_i)}.\]
	We extend the definition of $\vect{p}$ to a partial function on $H_k(1)$, as follows: given input $(w_1, \dots, w_k)$, let $S \subseteq [k]$ be the set of indices $i$ such that $w_i < 0$ and $T \subseteq [k]$ be the set of indices $i$ such that $w_i > 0$. We define $\vect{p}(w_1, \dots, w_k)$ to be the $\vect{q} \in \mathcal{D}$ such that
	\[(\vect{q}, 1) \oplus \parens{\bigoplus_{i \in S} (\vect{p}_i, -w_i)} = \bigoplus_{i \in T} (\vect{p}_i, w_i).\]
	Note that $\vect{q}$ is not guaranteed to exist, which is why we call $\vect{p}$ a partial function. However, if $\vect{q}$ exists then it is unique, by the subtraction axiom.
\end{defin}

We can now state the full axiomatization, including the monotonicity axiom.

\begin{defin}[Axioms for pooling operators] \label{def:wop}
	For a pooling operator $\oplus$ on $\mathcal{D}$, we define the following axioms.
	\begin{enumerate}
		\item \textbf{Weight additivity}: $\wt(\Pi_1 \oplus \Pi_2) = \wt(\Pi_1) + \wt(\Pi_2)$ for every $\Pi_1, \Pi_2$.
		\item \textbf{Commutativity}: $\Pi_1 \oplus \Pi_2 = \Pi_2 \oplus \Pi_1$ for every $\Pi_1, \Pi_2$.
		\item \textbf{Associativity}: $\Pi_1 \oplus (\Pi_2 \oplus \Pi_3) = (\Pi_1 \oplus \Pi_2) \oplus \Pi_3$ for every $\Pi_1, \Pi_2, \Pi_3$.
		\item \textbf{Continuity}: For every positive integer $k$ and $\vect{p}_1, \dots, \vect{p}_k$, the quantity\footnote{The continuity axiom is only well-defined conditioned on $\oplus$ being associative, which is fine for our purposes. We allow a proper subset of weights to be zero by defining the aggregate to ignore forecasts with weight zero.}
		\[\prb \parens{\bigoplus_{i = 1}^k (\vect{p}_i, w_i)}\]
		is a continuous function of $(w_1, \dots, w_k)$ on $\RR_{\ge 0}^k \setminus \{\vect{0}\}$.
		\item \textbf{Idempotence}: For every $\Pi_1$ and $\Pi_2$, if $\prb(\Pi_1) = \prb(\Pi_2)$ then $\prb(\Pi_1 \oplus \Pi_2) = \prb(\Pi_1)$.
		\item \textbf{Subtraction}: If $\Pi_1 \oplus \Pi_2 = \Pi_1 \oplus \Pi_3$ then $\Pi_2 = \Pi_3$.
		\item \textbf{Monotonicity}: There exist vectors $\vect{p}_1, \dots, \vect{p}_n \in \mathcal{D}$ such that $\vect{p}$ (as in Definition~\ref{def:p}) is a strictly cyclically monotone function from its domain to $\RR^n$.
	\end{enumerate}
\end{defin}

This monotonicity axiom essentially extends our previous monotonicity axiom (in Definition~\ref{def:wop_n2}) to a multi-dimensional setting. It states that there are $n$ ``anchor points" in $\mathcal{D}$ such that the function $\vect{p}$ from weight vectors to $\mathcal{D}$ that pools the anchor points with the weights given as input obeys a notion of monotonicity (namely cyclical monotonicity). Informally, this means that the vector of weights that one would need to give to the anchor points in order to arrive at a forecast $\vect{p}$ ``correlates" with the forecast $\vect{p}$ itself.\\

We now state the main theorem of our axiomatization.

\begin{theorem} \label{thm:representation}
	A pooling operator satisfies the axioms in Definition~\ref{def:wop} if and only if it is a QA pooling operator as in Definition~\ref{def:qa_ax}.\footnote{Recall that Definition~\ref{def:qa_ax} is narrower than Definition~\ref{def:qa_pooling}, since it excludes QA pools with respect to scoring rules that do not have convex exposure. Without convex exposure, the associativity, subtraction, and monotonicity axioms may be violated. For example, consider $\Pi_i = (\delta_i, 1)$ for $i = 1, 2, 3$, for the scoring rule with expected reward function $G(\vect{x}) = x_1^4 + x_2^4 + x_3^4$. We have that $(\Pi_1 \oplus_\vect{g} \Pi_2) \oplus_\vect{g} \Pi_3 = ((1/2, 1/2, 0), 2) \oplus_\vect{g} (\delta_3, 1) \approx ((0.182, 0.182, 0.635), 3)$, whereas $\Pi_1 \oplus_\vect{g} (\Pi_2 \oplus_\vect{g} \Pi_3) = ((0.635, 0.182, 0.182), 3)$}.
\end{theorem}

\begin{proof}
	We begin by noting the following fact, which follows from results in \cite[\S26]{roc70_book}.
	
	\begin{prop} \label{prop:g_inverse}
		A strictly cyclically monotone function $\vect{g}: \mathcal{D} \to \RR^n$ is injective, and its inverse $\vect{g}^{-1}$ is strictly cyclically monotone and continuous.\footnote{Why can't we apply this result again to $\vect{g}^{-1}$ to conclude that $\vect{g}$ is continuous, even though we did not assume it to be? The reason is that the proof of continuity relies on the convexity of $\mathcal{D}$; if $\vect{g}$ is discontinuous then the domain of $\vect{g}^{-1}$ may not be convex (or even connected), so we cannot apply the result to $\vect{g}^{-1}$.}
	\end{prop}
	
	We provide a partial proof below; it relies on the following observation.
	
	\begin{remark} \label{remark:cyc_mon_equiv}
		We can instead write the condition as
		\[\sum_{i = 1}^k (\vect{g}(\vect{x}_i) - \vect{g}(\vect{x}_{i - 1})) \cdot \vect{x}_i \ge 0.\]
		This is equivalent to the condition in Definition~\ref{def:cyc_mon}, because it is the same statement (with rearranged terms) when the $\vect{x}_i$'s are listed in reverse order.
	\end{remark}
	
	\begin{proof}
		First, suppose that $\vect{g}(\vect{x}) = \vect{g}(\vect{y})$. Then
		\[\vect{g}(\vect{x})(\vect{x} - \vect{y}) + \vect{g}(\vect{y})(\vect{y} - \vect{x}) = 0.\]
		Since $\vect{g}$ is \emph{strictly} cyclically monotone, this implies that $\vect{x} = \vect{y}$. (Note that we only use two-cycle monotonicity.)
		
		We now show that $\vect{g}^{-1}$ is strictly cyclically monotone. That is, we wish to show that
		\[\sum_{i = 1}^k \vect{x}_i \cdot (\vect{g}^{-1}(\vect{x}_i) - \vect{g}^{-1}(\vect{x}_{i - 1})) > 0\]
		for any $\vect{x}_1, \dots, \vect{x}_k = \vect{x}_0$ that are not all the same. (See Remark~\ref{remark:cyc_mon_equiv}.) By the cyclical monotonicity of $\vect{g}$, we have that
		\[\sum_{i = 1}^k \vect{g}(\vect{p}_i) \cdot (\vect{p}_i - \vect{p}_{i - 1}) > 0\]
		(the strictness of the inequality follows by the injectivity of $\vect{g}$: if $\vect{x}_i \neq \vect{x}_j$ then $\vect{p}_i \neq \vect{p}_j$). This means that
		\[\sum_{i = 1}^k \vect{x}_i \cdot (\vect{g}^{-1}(\vect{x}_i) - \vect{g}^{-1}(\vect{x}_{i - 1})) > 0,\]
		as desired. As for continuity, we defer to \cite[Theorem 26.5]{roc70_book}.
	\end{proof}
	
	Back to the proof of Theorem~\ref{thm:representation}, we first prove that any such $\oplus_{\vect{g}}$ satisfies the stated axioms. Weight additivity, commutativity, associativity, and idempotence are clear. Continuity follows from the formula
	\[\prb((\vect{p}_1, w_1) \oplus_{\vect{g}} (\vect{p}_2, w_2)) = \vect{g}^{-1} \parens{\frac{w_1 \vect{g}(\vect{p}_1) + w_2 \vect{g}(\vect{p}_2)}{w_1 + w_2}},\]
	noting that $\vect{g}^{-1}$ is continuous by Proposition~\ref{prop:g_inverse}. Likewise, subtraction follows from the fact that $\vect{g}$ is injective (by Proposition~\ref{prop:g_inverse}), as is $\vect{g}^{-1}$ (likewise). Monotonicity remains.
	
	The range of $\vect{g}$ contains an open subset\footnote{This follows from the invariance of domain theorem, which states that the image of an open subset of a manifold under an injective continuous map is open.} of $H_n(0)$, so in particular it contains the vertices of some translated and dilated copy of the standard simplex. That is, there are $n$ points $\vect{x}_1, \dots, \vect{x}_n$ in the range of $\vect{g}$ for which there is a positive scalar $a$ and vector $\vect{b}$ such that $a \delta_i + \vect{b} = \vect{x}_i$ for every $i$. (Here $\delta_i$ is the $i$-th standard basis vector in $\RR^n$.) We will let $\vect{p}_i$ be the pre-image of $\vect{x}_i$ under $\vect{g}$, so that $\vect{g}(\vect{p}_i) = a \delta_i + \vect{b}$.
	
	Observe that for any $\vect{w}$ in the domain of $\vect{p}$, we have
	\[\vect{g}(\vect{p}(\vect{w})) = \sum_{i = 1}^n w_i \vect{g}(\vect{p}_i) = \sum_{i = 1}^n w_i(a \delta_i + \vect{b}) = a\vect{w} + \vect{b},\]
	so
	\[\vect{p}(\vect{w}) = \vect{g}^{-1}(a\vect{w} + \vect{b}).\]
	We have that $\vect{g}^{-1}$ is strictly cyclically monotone (by Proposition~\ref{prop:g_inverse}), and it is easy to verify that for any strictly cyclically monotone function $\vect{f}$ and any $a > 0$ and $\vect{b}$, $\vect{f}(a\vect{x} + \vect{b})$ is a strictly cyclically monotone function of $\vect{x}$. Therefore, $\vect{p}(\vect{w}) = \vect{g}^{-1}(a\vect{w} + \vect{b})$ is strictly cyclically monotone, as desired.\\
	
	Now we prove the converse. Assume that we have a pooling operator $\oplus$ satisfying the axioms in Definition~\ref{def:wop}. We wish to show that $\oplus$ is $\oplus_{\vect{g}}$ for some $\vect{g}: \mathcal{D} \to H_n(0)$.
	
	For the remainder of this proof, let $\vect{p}_1, \dots, \vect{p}_n$ be vectors certifying the monotonicity of $\oplus$, and let $\vect{p}(\cdot)$ be as in Definition~\ref{def:p}.
	
	For any $\vect{q} \in \mathcal{D}$, let $\vect{g}(\vect{q}) := \vect{w} - \frac{1}{n} \vect{1}_n$, where $\vect{w} \in H_n(1)$ is such that $\vect{p}(\vect{w}) = \vect{q}$ and $\vect{1}_n$ is the all-ones vector. This raises the question of well-definedness: does this $\vect{w}$ necessarily exist, and if so, is it unique? The following claim shows that this is indeed the case.
	
	\begin{claim} \label{claim:p_bijective}
		The function $\vect{p}$, from the subset of $H_n(1)$ where it is defined to $\mathcal{D}$, is bijective.
	\end{claim}
	
	\begin{proof}
		The fact that $\vect{p}$ is injective follows from the fact that it is strictly cyclically monotone (see Proposition~\ref{prop:g_inverse}). We now show that $\vect{p}$ is surjective.
		
		Let $\vect{q} \in \mathcal{D}$. Define the function $\tilde{\vect{p}}: \Delta^{n + 1} \to \mathcal{D}$ by
		\[\tilde{\vect{p}}(w_1, \dots, w_{n + 1}) := \prb \parens{\parens{\bigoplus_{i = 1}^n (\vect{p}_i, w_i)} \oplus (\vect{q}, w_{n + 1})}.\]
		Since $\tilde{\vect{p}}$ is a continuous map\footnote{By the continuity axiom; here we use the more generalized form we stated earlier.} from $\Delta^{n + 1}$ (an $n$-dimensional manifold) to $\mathcal{D}$ (an $(n - 1)$-dimensional manifold), $\tilde{\vect{p}}$ is not injective.\footnote{This follows e.g. from the Borsuk-Ulam theorem.} So in particular, let $\vect{w}_1 \neq \vect{w}_2 \in \Delta^{n + 1}$ be such that $\vect{\tilde{p}}(\vect{w}_1) = \vect{\tilde{p}}(\vect{w}_2)$. That is, we have
		
		\begin{equation} \label{eq:ptilde_not_injective}
			\parens{\bigoplus_{i = 1}^n (\vect{p}_i, w_{1, i})} \oplus (\vect{q}, w_{1, n + 1}) = \parens{\bigoplus_{i = 1}^n (\vect{p}_i, w_{2, i})} \oplus (\vect{q}, w_{2, n + 1}).
		\end{equation}
		
		Observe that $w_{1, n + 1} \neq w_{2, n + 2}$; for otherwise it would follows from the subtraction axiom that two different combinations of the $\vect{p}_i$'s would give the same probability, contradicting the fact that $\vect{p}$ is injective. Without loss of generality, assume that $w_{1, n + 1} > w_{2, n + 1}$. We can rearrange the terms in Equation~\ref{eq:ptilde_not_injective} to look as follows.
		\[(\vect{q}, w_{1, n + 1} - w_{2, n + 1}) \oplus \parens{\bigoplus_{i \in S} (\vect{p}_i, v_i)} = \bigoplus_{i \in T \subseteq [n] \setminus S} (\vect{p}_i, v_i)\]
		for some positive $v_1, \dots, v_n$. By the distributive property, we may multiply all weights by $\frac{1}{w_{1, n + 1} - w_{2, n + 1}}$. The result will be an equation as in Definition~\ref{def:p}, certifying that $\vect{q}$ is in the range of the function $\vect{p}$, as desired.
	\end{proof}
	
	We return to our main proof, now that we have shown that our function $\vect{g}(\vect{q}) := \vect{w} - \frac{1}{n} \vect{1}$, where $\vect{w} \in H_n(1)$ is such that $\vect{p}(\vect{w}) = \vect{q}$, is well-defined. In fact, we can simply write $\vect{g}(\vect{q}) = \vect{p}^{-1}(\vect{q}) - \frac{1}{n} \vect{1}$. (The vector $\frac{1}{n} \vect{1}$ is fairly arbitrary; it only serves the purpose of forcing the range of $\vect{g}$ to lie in $H_n(0)$ instead of $H_n(1)$.)
	
	We first show that the equation that defines $\oplus_{\vect{g}}$ holds -- that is, that if $(\vect{q}_1, v_1) \oplus (\vect{q}_2, v_2) = (\vect{q}, v_1 + v_2)$ (with $v_1, v_2 \ge 0$, not both zero), then
	\[\vect{g}(\vect{q}) = \frac{v_1 \vect{g}(\vect{q}_1) + v_2 \vect{g}(\vect{q}_2)}{v_1 + v_2}.\]
	Let $\vect{w}_1, \vect{w}_2 \in H_n(1)$ be such that $\vect{q}_1 = \vect{p}(\vect{w}_1)$ and $\vect{q}_2 = \vect{p}(\vect{w}_2)$. It is intuitive that $\vect{q} = \vect{p} \parens{\frac{v_1 \vect{w_1} + v_2 \vect{w_2}}{v_1 + v_2}}$, but we show this formally.
	
	\begin{claim} \label{claim:p_linear}
		Given $\vect{q}_1, \vect{q}_2 \in \mathcal{D}$ with $\vect{q}_1 = \vect{p}(\vect{w}_1), \vect{q}_2 = \vect{p}(\vect{w}_2)$, and $0 \le \alpha \le 1$, we have
		\[\vect{p}(\alpha \vect{w}_1 + (1 - \alpha) \vect{w}_2) = (\vect{q}_1, \alpha) \oplus (\vect{q}_2, 1 - \alpha).\]
	\end{claim}
	
	\begin{proof}
		Note that
		\begin{align*}
			(\vect{q}_1, 1) \oplus \parens{\bigoplus_{i: w_{1, i} < 0} (\vect{p}_i, w_{1, i})} &= \bigoplus_{i: w_{1, i} > 0} (\vect{p}_i, w_{1, i})\\
			(\vect{q}_2, 1) \oplus \parens{\bigoplus_{i: w_{2, i} < 0} (\vect{p}_i, w_{2, i})} &= \bigoplus_{i: w_{2, i} > 0} (\vect{p}_i, w_{2, i}).
		\end{align*}
		Applying the distributive property to the two above equations with constants $\alpha$ and $1 - \alpha$, respectively, and adding them, we get that
		\begin{align*}
			&\parens{\vect{q}_1, \alpha} \oplus \parens{\vect{q}_2, 1 - \alpha} \oplus \parens{\bigoplus_{i: w_{1, i} < 0} (\vect{p}_i, \alpha w_{1, i})} \oplus \parens{\bigoplus_{i: w_{2, i} < 0} (\vect{p}_i, (1 - \alpha) w_{2, i})}
			\\&= \parens{\bigoplus_{i: w_{1, i} > 0} (\vect{p}_i, \alpha w_{1, i})} \oplus \parens{\bigoplus_{i: w_{2, i} > 0} (\vect{p}_i, (1 - \alpha) w_{2, i})}.
		\end{align*}
		We have that $\parens{\vect{q}_1, \alpha} \oplus \parens{\vect{q}_2, 1 - \alpha} = (\vect{q}, 1)$. It follows (after rearranging terms, from Definition~\ref{def:p}) that $\vect{q} = \vect{p} \parens{\frac{v_1 \vect{w_1} + v_2 \vect{w_2}}{v_1 + v_2}}$.
	\end{proof}
	
	Applying Claim~\ref{claim:p_linear} with $\alpha = \frac{v_1}{v_1 + v_2}$, we find that
	\[\vect{g}(\vect{q}) = \frac{v_1 \vect{w_1} + v_2 \vect{w_2}}{v_1 + v_2} - \frac{1}{n} \vect{1} = \frac{v_1 \parens{\vect{g}(\vect{q}_1) + \frac{1}{n} \vect{1}} + v_2\parens{\vect{g}(\vect{q}_2) + \frac{1}{n} \vect{1}}}{v_1 + v_2} - \frac{1}{n} \vect{1} = \frac{v_1 \vect{g}(\vect{q}_1) + v_2 \vect{g}(\vect{q}_2)}{v_1 + v_2},\]
	as desired.\\
	
	It remains to show that $\vect{g}$ is continuous, strictly cyclically monotone, and has convex range. By the monotonicity axiom, $\vect{p}$ is strictly cyclically monotone. It follows by Proposition~\ref{prop:g_inverse} that its inverse its continuous and strictly cyclically monotone. Therefore, $\vect{g}$ is continuous and cyclically monotone (as it is simply a translation of $\vect{p}^{-1}(\vect{q})$ by $\frac{1}{n} \vect{1}$).
	
	Finally, to show that $\vect{g}$ has convex range, we wish to show that $\vect{p}^{-1}$ has convex range; or, in other words, that the domain on which $\vect{p}$ is defined is convex. And indeed, this follows straightforwardly from Claim~\ref{claim:p_linear}. Let $\vect{w}_1, \vect{w}_2$ be in the domain of $\vect{p}$, with $\vect{p}(\vect{w}_1) = \vect{q}_1, \vect{p}(\vect{w}_2) = \vect{q}_2$. Then for any $0 \le \alpha \le 1$, we have that
	\[\vect{p}(\alpha \vect{w}_1 + (1 - \alpha) \vect{w}_2) = (\vect{q}_1, \alpha) \oplus (\vect{q}_2, 1 - \alpha),\]
	so in particular $\alpha \vect{w}_1 + (1 - \alpha) \vect{w}_2$ is in the domain of $\vect{p}$. This concludes the proof.
\end{proof}

%% file: Appendix_Convex_Exposure.tex
\section{The convex exposure property} \label{appx:convex_exposure}
Several of our results have been contingent on the convex exposure property. In this e-companion, we consider when the convex exposure property holds. Our first result is that it always holds in the case of a binary outcome (i.e.\ $n = 2$).

\begin{prop}\leavevmode \label{prop:n2_convex_exposure}
	If $n = 2$, every (continuous) proper scoring rule has convex exposure.
\end{prop}

\begin{proof}
	Consider a proper scoring rule $s$ with forecast domain $\mathcal{D}$. Since $\mathcal{D}$ is connected and $\vect{g}$ is continuous on $\mathcal{D}$, the range of $\vect{g}$ over $\mathcal{D}$ is connected. In the $n = 2$ outcome case, the range of $\vect{g}$ lies on the line $\{(x_1, x_2): x_1 + x_2 = 0\}$, and a connected subset of a line is convex.
\end{proof}

As we shall see, the convex exposure property holds for nearly all of the most commonly used scoring rules even in higher dimensions.

(A note on notation: in this section we use $p_j$ instead of $p(j)$ to refer to the $j$-th coordinate of a probability distribution $\vect{p}$.)\\

We now show that scoring rules that --- like the logarithmic scoring rule --- ``go off to infinity" have convex exposure.

\begin{prop} \label{prop:steep_convex}
Let $s$ be a proper scoring rule whose forecast domain is the interior of $\Delta^n$, such that for any point $\vect{x}$ on the boundary of $\Delta^n$, and for any sequence $\vect{x}_1, \vect{x}_2, \dots$ converging to $\vect{x}$, $\lim_{k \to \infty} \norm{\vect{g}(\vect{x}_k)}_2 = \infty$.\footnote{\cite{acv13} say that $G$ is a \emph{pseudo-barrier function} if this condition is satisfied.} Then $s$ has convex exposure.
\end{prop}

This is a statement of convex analysis --- namely that if $\norm{\vect{g}}_2$ approaches $\infty$ on the boundary of a convex set, then the range of $\vect{g}$ is convex (assuming $\vect{g}$ is the gradient of a differentiable convex function). We refer the reader to \cite[Theorem 26.5]{roc70_book} for the proof. In non-pathological cases, the basic intuition is that \emph{every} $\vect{v} \in \{\vect{x}: \sum_i x_i = 0\}$ is the gradient of $G$ at some point. In these cases, $\nabla G(\vect{x}) = \vect{v}$ where $\vect{x}$ minimizes $G(\vect{v}) - \vect{v} \cdot \vect{x}$; the $\lim_{k \to \infty} \norm{g(\vect{x}_k)}_2 = \infty$ condition means that this minimum does not occur on the boundary of $\Delta^n$.

\begin{corollary} \label{cor:steep}
The following scoring rules have convex exposure:
\begin{itemize}
\item The logarithmic scoring rule.
\item The scoring rule given by $G(\vect{p}) = -\sum_j p_j^{\gamma}$ for $\gamma \in (0, 1)$.
\item The scoring rule given by $G(\vect{p}) = -\sum_j \ln p_j$, which can be thought of as the limit of the $G$ in the previous bullet point as $\gamma \to 0$.\footnote{This is a natural way to think of this scoring rule because $\nabla G(\vect{p}) = -(p_1^{-1}, \dots, p_n^{-1})$.}
\item The scoring rule $hs$ given by $G_{hs}(\vect{p}) = -\prod_j p_j^{1/n}$.
\end{itemize}
\end{corollary}

\paragraph{\textbf{The $hs$ scoring rule}} The last of these scoring rules is a generalization of the scoring rule $hs(q) = 1 - \sqrt{\frac{1 - q}{q}}$ used in \cite{bb20} as part of proving their minimax theorem for randomized algorithms.\footnote{Here we are using the shorthand notation for the $n = 2$ outcome case discussed in Remark~3.5 of the main article.} The authors used this scoring rule as a key ingredient in their minimax theorem for randomized algorithms. The key property of the scoring rule was a result about its \emph{amplification} \cite[Lemma 3.10]{bb20}. The authors define a \emph{forecasting algorithm} to be a generalization of a randomized algorithm that outputs an estimated \emph{probability} that an output should be accepted. Then, roughly speaking, the authors show that given a forecasting algorithm $R$, it is possible to create a forecasting algorithm $R'$ that has a much larger expected score from the scoring rule $hs$ by combining running $R$ a small number of times and combining the outputs. This is an important new result in theoretical computer science and suggests that $hs$ deserves more attention.

Since additive and multiplicative constants are irrelevant, we may treat $hs(q) = -\frac{1}{2} \sqrt{\frac{1 - q}{q}}$. Observe that (in the case of two outcomes), the expected score $G_{hs}$ on a report of $q$ is
\[G_{hs}(q) = q \parens{-\frac{1}{2} \sqrt{\frac{1 - q}{q}}} + (1 - q) \parens{-\frac{1}{2} \sqrt{\frac{1 - q}{q}}} = -\sqrt{q(1 - q)}.\]
That is, $G_{hs}$ is precisely negative the geometric mean of $q$ and $1 - q$. This motivates us to generalize $hs$ to a setting with $n$ outcomes by setting
\[G_{hs}(\vect{p}) := -\prod_{i = 1}^n p_i^{1/n}.\]
It should not be obvious that this function is convex, but it turns out to be; this is the precise statement of an inequality known as Mahler's inequality \cite{wiki:mahler}.\\

Next we note that the quadratic scoring rule has convex exposure, since its exposure function $\vect{g}(\vect{p}) = 2\vect{p}$ (modulo $\vect{1}_n$ as discussed in Remark~3.12 of the main article) maps any convex set to a convex set.
\begin{prop}
The quadratic scoring rule has convex exposure.
\end{prop}

\emph{Spherical} scoring rules --- the third most studied proper scoring rules, after the quadratic and logarithmic rules --- also have convex exposure.
\begin{defin}[Spherical scoring rules]
\cite[Example 2]{gr07}
For any $\alpha > 1$, define the \emph{spherical scoring rule with parameter $\alpha$} to be the scoring rule given by
\[G_{\text{sph}, \alpha}(\vect{p}) := \parens{\sum_{i = 1}^n p_i^{\alpha}}^{1/\alpha}.\]
If the ``spherical scoring rule" is referenced with no parameter $\alpha$ given, $\alpha$ is presumed to equal $2$.
\end{defin}

\begin{prop} \label{prop:sph_convex}
For any $\alpha > 1$, the spherical scoring rule with parameter $\alpha$ has convex exposure.
\end{prop}

\begin{proof}
Fix $\alpha > 1$. We will write $G$ in place of $G_{\text{sph}, \alpha}$. We have\footnote{As discussed in Remark~3.12 of the main article, the range of $\vect{g}$ should be thought of as modulo $T(\vect{1}_n)$. However, we find it convenient for this proof to think of it as lying in $\RR^n$ and project later.}
\begin{equation} \label{eq:g_sph}
\vect{g}(\vect{p}) = \parens{\sum_{j = 1}^n p_j^{\alpha}}^{(1/\alpha) - 1} (p_1^{\alpha - 1}, \dots, p_n^{\alpha - 1}).
\end{equation}
Now, define the \emph{$n$-dimensional unit $\beta$-sphere}, i.e. $\{\vect{x}: \sum_j x_j^{\beta} = 1\}$, and define the \emph{$n$-dimensional unit $\beta$-ball} correspondingly (i.e. with $\le$ in place of $=$). The range of $\vect{g}$ is precisely the part of the $n$-dimensional unit $\frac{\alpha}{\alpha - 1}$-sphere with all non-negative coordinates. Indeed, on the one hand, for any $\vect{p}$ we have
\[\sum_j g_j(\vect{p})^{\alpha/(\alpha - 1)} = \parens{\sum_j p_j^\alpha}^{-1} \cdot \sum_j p_j^\alpha = 1\]
(where $g_j(\vect{p})$ denotes the $j$-th coordinate of $\vect{g}(\vect{p})$ as in Equation~\ref{eq:g_sph}). On the other hand, given a point $\vect{x}$ on the unit $\frac{\alpha}{\alpha - 1}$-sphere with all non-negative coordinates,
\[\vect{p} = \parens{\sum_j x_j^{1/(\alpha - 1)}}^{-1} \parens{x_1^{1/(\alpha - 1)}, \dots, x_n^{1/(\alpha - 1)}}\]
lies in $\Delta^n$ and satisfies $\vect{g}(\vect{p}) = \vect{x}$.

The crucial point for us is that for $\beta > 1$, the unit $\beta$-ball is convex. This means that for any such $\beta$, the convex combination of any number of points on the unit $\beta$-sphere will lie in the unit $\beta$-ball. Since $\frac{\alpha}{\alpha - 1} > 1$ for $\alpha > 1$, we have that for arbitrary $\vect{p}, \vect{q} \in \Delta^n$ and $w \in [0, 1]$, $w \vect{g}(\vect{p}) + (1 - w) \vect{g}(\vect{q})$ lies in the unit $\beta$-ball --- in fact, in the part with all non-negative coordinates. Now, consider casting a ray from this convex combination point in the positive $\vect{1}_n$ direction. All points on this ray are equivalent to this point modulo $T(\vect{1}_n)$, and this ray will intersect the unit $\beta$-sphere at some point $\vect{x}$ with all non-negative coordinates. The point $\vect{p} \in \Delta^n$ with $\vect{g}(\vect{p}) = \vect{x}$ satisfies
\[\vect{g}(\vect{p}) = \vect{g}(\vect{p}) + (1 - w) \vect{g}(\vect{q}).\]
This completes the proof.
\end{proof}

\begin{remark}
The above proof gives a geometric interpretation of the QA pooling with respect to the spherical scoring rule, particularly for $\alpha = 2$. In the $\alpha = 2$ case, pooling amounts to taking the following steps:
\begin{enumerate}[label=(\arabic*)]
\item Scale each forecast so it lies on the unit sphere.
\item Take the weighted average of the resulting points in $\RR^n$.
\item Shift the resulting point in the positive $\vect{1}_n$ direction to the unique point in that direction that lies on the unit sphere.
\item Scale this point so that its coordinates add to $1$.
\end{enumerate}
\end{remark}

Finally we consider the parametrized family known as \emph{Tsallis scoring rules}, defined in \cite{tsa88}.
\begin{defin}[Tsallis scoring rules]
For $\gamma > 1$, the \emph{Tsallis scoring rule with parameter $\gamma$} is the rule given by
\[G_{\text{Tsa},\gamma}(\vect{p}) = \sum_{j = 1}^m p_j^\gamma.\]
\end{defin}
Setting $\gamma = 2$ above yields the quadratic scoring rule. Note also that we have already addressed the scoring rule given by $G(\vect{p}) = \pm \sum_j p_j^\gamma$ for $\gamma \le 1$ (except $\gamma = 0, 1$, which are degenerate), with the sign chosen to make $G$ convex: these scoring rules have convex exposure by Proposition~\ref{prop:steep_convex}. The following proposition completes our analysis for this natural class of scoring rules.

\begin{prop} \label{prop:tsallis}
For $\gamma \le 2$, the Tsallis scoring rule with parameter $\gamma$ has convex exposure. For $\gamma > 2$, this is not the case if $n > 2$.
\end{prop}

\begin{proof}[Proof of Proposition~\ref{prop:tsallis}]
Fix $\gamma > 1$. We will write $G$ in place of $G_{\text{Tsa},\gamma}$. Up to a multiplicative factor of $\gamma$ that we are free to ignore, we have
\[\vect{g}(\vect{p}) = (p_1^{\gamma - 1}, \dots, p_n^{\gamma - 1}).\]
Let $\vect{p}, \vect{q} \in \Delta^n$ and $w \in [0, 1]$. We wish to find an $\vect{x} \in \Delta^n$ such that $\vect{g}(\vect{x}) = w \vect{g}(\vect{p}) + (1 - w) \vect{g}(\vect{q})$, i.e.
\[w p_j^{\gamma - 1} + (1 - w) q_j^{\gamma - 1} + c = x_j^{\gamma - 1},\]
for all $j \in [n]$, for some $c$. Since $\sum_j x_j = 1$, this $c$ must satisfy
\begin{equation} \label{eq:tsallis_requirement}
\sum_j (w p_j^{\gamma - 1} + (1 - w) q_j^{\gamma - 1} + c)^{1/(\gamma - 1)} = 1.
\end{equation}
Let $h(x) := \sum_j (w p_j^{\gamma - 1} + (1 - w) q_j^{\gamma - 1} + x)^{1/(\gamma - 1)}$. Note that $h$ is increasing in $x$.

First consider the case that $\gamma \le 2$. By concavity, we have that $w p_j^{\gamma - 1} + (1 - w) q_j^{\gamma - 1} \le (wp_j + (1 - w)q_j)^{\gamma - 1}$. This means that
\[h(0) = \sum_j (w p_j^{\gamma - 1} + (1 - w) q_j^{\gamma - 1})^{1/(\gamma - 1)} \le \sum_j (wp_j + (1 - w)q_j) = 1.\]
On the other hand, $\lim_{x \to \infty} h(x) = \infty$. Since $h$ is continuous, there must be some $x \in [0, \infty)$ such that $h(x) = 1$; call this value $c$. Then let
\[x_j = (w p_j^{\gamma - 1} + (1 - w) q_j^{\gamma - 1} + c)^{1/(\gamma - 1)}.\]
Then every $x_j$ is nonnegative and $\sum_j x_j = 1$, so we have succeeded.

Now consider the case that $\gamma > 2$, and consider as a counterexample $\vect{p} = (1, 0, \dots, 0)$, $\vect{q} = (0, 1, 0, \dots, 0)$, and $w = \frac{1}{2}$. To satisfy Equation~\ref{eq:tsallis_requirement}, we are looking for $c$ such that
\[h(c) = 2 \parens{\frac{1}{2} + c}^{1/(\gamma - 1)} + (n - 2)c^{1/(\gamma - 1)} = 1.\]
Note that $h(0) = 2 \cdot 2^{-1/(\gamma - 1)} = 2^{(\gamma - 2)/(\gamma - 1)} > 1$, so $c < 0$ (as $h$ is increasing). But in that case $x_j^{\gamma - 1} < 0$ for any $j \ge 3$, a contradiction (assuming $n > 2$).
\end{proof}

Note that because $\nabla G_{\text{Tsa}, \gamma}(\vect{p}) = (p_1^{\gamma - 1}, \dots, p_n^{\gamma - 1})$ (up to a constant factor), QA pooling with respect to the Tsallis scoring rule  can be thought of as an appropriately scaled coordinate-wise $(\gamma - 1)$-power mean. For $\gamma = 2$ it is the coordinate-wise arithmetic average. For $\gamma = 3$ it is the coordinate-wise root mean square, but with the average of the squares scaled by an appropriate additive constant so that, upon taking the square roots, the probabilities add to $1$. (However, as the Tsallis score with parameter $3$ does not have convex exposure, this is not always well-defined.)

In Corollary~\ref{cor:steep} we mentioned that the scoring rule given by $G(\vect{p}) = -\sum_j \ln p_j$ can be thought of as an extension to $\gamma = 0$ of (what we are now calling) the Tsallis score, because the derivative of $\ln x$ is $x^{-1}$. QA pooling with respect to this scoring rule is, correspondingly, the $-1$-power mean, i.e. the harmonic mean. This pooling method is appropriately referred to as \emph{harmonic pooling}, see e.g. \cite[\S4.2]{ddmc95}.

Finally, we note that the logarithmic scoring rule can likewise be thought of as an extension of the Tsallis score to $\gamma = 1$, in that the second derivative of $x \ln x$ is $x^{-1}$. It is likewise natural to call the geometric mean the $0$-power mean; notice that logarithmic pooling is precisely an appropriately scaled coordinate-wise geometric mean.